\documentclass[letterpaper,12pt]{article}
\usepackage{url}
\usepackage{latexsym}
\usepackage[boxruled,vlined,nokwfunc]{algorithm2e}
\usepackage{amsmath}
\usepackage{amsthm}
\usepackage{amssymb}
\usepackage{amsfonts}
\usepackage{wrapfig}
\usepackage{enumitem}
\usepackage[usenames,dvipsnames]{color}
\usepackage[pdftitle={A forward-backward single-source shortest paths algorithm},
            pdfauthor={David B. Wilson and Uri Zwick},
            bookmarks=true,bookmarksopen=true,bookmarksopenlevel=3,unicode]{hyperref}
\usepackage[all]{hypcap}

\newcommand{\old}[1]{}

\newtheorem{theorem}{Theorem}[section]
\newtheorem{lemma}[theorem]{Lemma}

\newtheorem{definition}[theorem]{Definition}

\newtheorem{remark}[theorem]{Remark}

\newcommand{\ee}{{\rm e}}
\newcommand{\E}{\mathbb{E}}
\renewcommand{\Pr}{\mathbb{P}}
\newcommand{\Var}{\operatorname{Var}}

\newcommand{\Eoutper}{E^{\mathrm{out}}_{\mathrm{per}}}
\newcommand{\Einper}{E^{\mathrm{in}}_{\mathrm{per}}}
\newcommand{\Eper}{E_{\mathrm{per}}}

\newcommand{\ignore}[1]{}
\newcommand{\remove}[1]{}
\newcommand{\reals}{\mathbb{R}}

\newcommand{\sssp}{\mbox{{\tt sssp}}}
\newcommand{\Spira}{\mbox{{\tt Spira}}}

\newcommand{\extractmin}{\mbox{{\tt extract}-{\tt min}}}
\newcommand{\decreasekey}{\mbox{{\tt decrease}-{\tt key}}}

\newcommand{\INSERT}{\mbox{{\tt insert}}}

\newcommand{\extmin}{\mbox{{\tt extract}-{\tt min}}}
\newcommand{\minkey}{\mbox{{\tt min}}}

\newcommand{\forward}{\mbox{{\tt forward}}}
\newcommand{\request}{\mbox{{\tt request}}}
\newcommand{\NEXT}{\mbox{\tt next}}
\newcommand{\nextout}{\mbox{{\tt next}-{\tt out}}}

\newcommand{\append}{\mbox{\tt append}}
\newcommand{\backward}{\mbox{{\tt backward}}}

\newcommand{\current}{\mbox{{\tt current}}}
\newcommand{\ADVANCE}{\mbox{{\tt advance}}}
\newcommand{\RESET}{\mbox{{\tt reset}}}

\newcommand{\FALSE}{{\tt false}}
\newcommand{\TRUE}{{\tt true}}
\newcommand{\NOT}{{\tt not}}
\newcommand{\AND}{{\tt and}}
\newcommand{\OR}{{\tt or}}

\newcommand{\EXP}{\operatorname{\textsc{Exp}}}
\newcommand{\Poi}{\operatorname{\textsc{Poisson}}}
\newcommand{\Geo}{\operatorname{\textsc{Geometric}}}
\newcommand{\NB}{\operatorname{\textsc{NB}}}
\newcommand{\BIN}{\operatorname{\textsc{B}}}

\renewcommand{\G}{{\cal K}}

\newcommand{\R}{\mathbb{R}}

\renewcommand{\bot}{\texttt{\textup{nil}}}
\def\clap#1{\hbox to 0pt{\hss#1\hss}}

\begin{document}

\title{\ \clap{A forward-backward single-source shortest paths algorithm
\footnote{An extended abstract of this article appeared in \textit{Proc.\ IEEE 54th Annual Symposium on Foundations of Computer Science}, pages 707--716, 2013.}
}\ }
\author{
\begin{tabular}{c}
\href{http://dbwilson.com}{David B. Wilson}\\
\small Microsoft Research
\end{tabular}
\and
\begin{tabular}{c}
\href{http://www.cs.tau.ac.il/~zwick/}{Uri Zwick}
\footnote{Work supported by grant no.\ 2012338 of the United States -- Israel Binational Science Foundation (BSF)
and by the Israeli Centers of Research Excellence (I-CORE) program (Center No.\ 4/11).}
\\
\small Tel Aviv University
\end{tabular}
}

\date{}

\maketitle

\begin{abstract}
We describe a new \emph{forward-backward\/} variant of Dijkstra's and
Spira's Single-Source Shortest Paths (SSSP) algorithms. While
essentially all SSSP algorithm only scan edges \emph{forward}, the new
algorithm scans some edges \emph{backward}.  The new algorithm assumes
that edges in the outgoing and incoming adjacency lists of the
vertices appear in non-decreasing order of weight. (Spira's algorithm
makes the same assumption about the outgoing adjacency lists, but
does not use incoming adjacency lists.)  The running time of the
algorithm on a complete directed graph on $n$ vertices with
independent exponential edge weights is $O(n)$, with very high
probability. This improves on the previously best result of $O(n\log n)$,
which is best possible if only forward scans are allowed, exhibiting
an interesting separation between forward-only and forward-backward
SSSP algorithms.  As a consequence, we also get a new all-pairs
shortest paths algorithm.  The expected running time of the algorithm
on complete graphs with independent exponential edge weights is
$O(n^2)$, matching a recent algorithm of Demetrescu and Italiano as
analyzed by Peres \textit{et al}.  Furthermore, the probability that
the new algorithm requires more than $O(n^2)$ time is
\emph{exponentially\/} small, improving on the $O(n^{-1/26})$
probability bound obtained by Peres \textit{et al.}
\end{abstract}

\section{Introduction}\label{S-intro}

\subsection{Shortest Paths}
The \emph{Single-Source Shortest Paths\/} (SSSP) problem, which calls for the computation of
a tree of shortest paths from a given vertex in a
directed or undirected graph with non-negative edge weights,
is one of the most important and most studied algorithmic graph problems. The classical algorithm of Dijkstra \cite{Di59}, implemented with an appropriate priority queue data structure, e.g., Fibonacci heaps (Fredman and Tarjan \cite{FrTa87}), solves the problem in $O(m+n\log n)$ time, where~$m$ is the number of edges and~$n$ is the number of vertices in the graph. For \emph{undirected\/} graphs with non-negative \emph{integer\/} edge weights, Thorup \cite{Thorup99} obtained an algorithm that runs in $O(m+n)$ time.

The running time of Dijkstra's algorithm is almost linear in the size of the input graph, and the running time of Thorup's algorithm is linear. Can we hope for a \emph{sublinear\/} time algorithm for the SSSP problem? In general the answer is of course ``no'', as an SSSP algorithm must examine essentially all the edges of the graph. There are, however, some interesting settings in which the input graph undergoes an initial \emph{preprocessing\/} phase after which it may be possible to solve the SSSP problem in sublinear time. We consider here a particularly simple such preprocessing phase that sorts the edges in the adjacency lists of the vertices of the graph in non-decreasing order of weight.

The \emph{All-Pairs Shortest Paths\/} (APSP) problem calls for the solution of the SSSP problem from every vertex of the input graph. It can clearly be solved in $O(mn+n^2\log n)$ time by running Dijkstra's algorithm from every vertex. Pettie \cite{Pettie04} improved this running time to $O(mn+n^2\log\log n)$. The problem can be solved in $O(mn)$ time, if the graph is undirected and the edge weights are integral, by running Thorup's \cite{Thorup99} algorithm from each vertex.

\subsection{Average case results}\label{sub:average}
Many authors considered the \emph{average case\/} complexity of the SSSP and APSP problems. Perhaps the simplest setting for such studies is the case of a complete directed graph on~$n$ vertices in which the weight of each edge is an independent \emph{exponential\/} random variable.
We denote this probabilistic model by $\G_n(\EXP(1))$.
Hassin and Zemel \cite{HaZe85} and Frieze and Grimmett \cite{FrGr85} gave simple algorithms that solve the APSP problem, when the input graph is drawn from $\G_n(\EXP(1))$, in $O(n^2\log n)$ expected time.

Spira \cite{Spira73} initiated the study of the expected running time of SSSP and APSP algorithms in a much more general probabilistic model, now referred to as the \emph{end-point independent\/} model.
The input graph in this model is a complete directed graph on~$n$ vertices.
Each vertex~$v$ has a (deterministic or stochastic) process that generates $n-1$ non-negative edge weights. These edge weights are randomly permuted and assigned to the outgoing edges of~$v$.
The process associated with each vertex is \emph{arbitrary}; different vertices may have different processes. Spira \cite{Spira73} gave an APSP algorithm whose expected running time in this model is $O(n^2\log^2 n)$. Spira's algorithm first applies the sorting preprocessing step described above and then solves each SSSP problem in $O(n\log^2 n)$ expected time.

Spira's result was improved by several authors. Takaoka and Moffat \cite{TaMo80} improved the running time to $O(n^2\log n\log\log n)$. Bloniarz \cite{Bloniarz83} improved it to $O(n^2\log n\log^*n)$. Finally, Moffat and Takaoka \cite{MoTa87} and Mehlhorn and Priebe \cite{MePr97} (see also recent simplifications by Takaoka and Hashim \cite{MR2783521,Takaoka12}) improved the running time to $O(n^2\log n)$. All these algorithms, like Dijkstra's and Spira's algorithm use only the \emph{outgoing\/} adjacency lists of the graph. They all start by sorting the outgoing adjacency lists and then running an SSSP algorithm from each vertex. The fastest algorithms above solve each SSSP problem in $O(n\log n)$ expected time. Mehlhorn and Priebe \cite{MePr97} showed that for the endpoint independent model, $\Omega(n\log n)$ expected time is best possible for algorithms that can only access the (sorted) outgoing adjacency lists of the graph.

Peres \textit{et al.} \cite{PSSZ} recently revisited the more basic $\G_n(\EXP(1))$ model, in which edge weights are independent exponential random variables.  They showed that Demetrescu and Italiano's dynamic APSP algorithm \cite{DeIt04a,DeIt06}, when used with bucket-based priority queues, has expected running time $O(n^2)$, which is clearly optimal.  They also showed that the running time is $O(n^2)$ except with probability $O(n^{-1/26})$ \cite{PSSZ}.  This algorithm is markedly different from all algorithms mentioned above as it does not sort the edge weights and then solve an SSSP problem from each vertex.  Rather, it finds all distances in the graph ``simultaneously''.

\subsection{A new SSSP algorithm}
Can the SSSP problem be solved in $O(n)$ time, assuming that the edge weights are independent exponential random variables and that the adjacency lists of the graph are given in sorted order?  Adapting the argument of Mehlhorn and Priebe \cite{MePr97} we show that any SSSP algorithm that can only access the (sorted) outgoing adjacency lists of the input graph, which is assumed to be drawn from the directed graph $\G_n(\EXP(1))$, must examine $\Omega(n\log n)$ edges, with high probability.
As it is not clear how incoming adjacency lists could be used to speed up SSSP algorithms, this seems to give a negative answer to the question.

Surprisingly, we show however that the SSSP problem on $\G_n(\EXP(1))$ \emph{can\/} be solved in $O(n)$ time,
with very high probability, beating the above $\Omega(n\log n)$ lower bound, by a \emph{forward-backward\/} algorithm that uses both the outgoing and incoming (sorted) adjacency lists of the input graph.
(If the graph is undirected, then of course it suffices to examine just the outgoing adjacency lists.)
 To the best of our knowledge, incoming adjacency lists were never used before by SSSP algorithms.
Although we only analyze our new algorithm in an ideal probabilistic model, we believe that suitable variants of the new algorithm may be used to speed up SSSP computations in more realistic settings.

We develop the new $O(n)$ time SSSP algorithm in two steps. In the first step, which is by far the more challenging step, and where most of the novelty in this paper lies, we devise and analyze an SSSP algorithm that \emph{scans\/} (or \emph{examines\/}) only $O(n)$ edges of the graph, with very high probability. On average, the algorithm examines a constant number of edges incident to each vertex. As explained, it is crucial here that the algorithm is allowed to examine both the outgoing and incoming adjacency lists of each vertex. In the second, and more standard step, we show that the algorithm can be implemented to run in $O(n)$ time.
On average, the algorithm is only allowed to perform a constant number of operations per edge examined. As essentially all Dijkstra-like SSSP algorithms, our new algorithm uses a \emph{priority-queue\/} data structure. To get an $O(n)$-time implementation we need to perform priority-queue operations in $O(1)$ expected amortized cost. We show that this is possible in our setting using relatively simple \emph{bucket-based\/} priority queues.

\subsection{A new APSP algorithm}
One setting in which the sortedness assumption of the adjacency lists may be justified is that of solving the APSP problem.
In the APSP setting, we can afford to spend $O(n^2)$ time on bucket-sorting the adjacency lists. We can then solve an SSSP problem from each vertex of the graph in $O(n)$ time, getting an $O(n^2)$-time algorithm for solving the APSP problem on $\G_n(\EXP(1))$. This matches the recent result of Peres \textit{et al.} \cite{PSSZ}. The new $O(n^2)$-time APSP algorithm is very different from the algorithm of Peres \textit{et al.} \cite{PSSZ}, which does not simply run an SSSP algorithm from each vertex. Furthermore, while the APSP algorithm of \cite{PSSZ} runs in $O(n^2)$ expected time, it was only shown in \cite{PSSZ} that the probability that it requires more than $O(n^2)$ time is $O(n^{-1/26})$. We show here, on the other hand, that the probability that our new algorithm requires more than $O(n^2)$-time is \emph{exponentially\/} small, specifically it is at most $\exp(-\Theta(n/\log n))$.

\subsection{On the probabilistic model}
For simplicity and concreteness, we stated all our results in the $\G_n(\EXP(1))$ probabilistic model. Exponential edge weights are convenient to work with due to their memoryless property.
However, the assumption that edge weights are drawn from an exponential distribution can be greatly relaxed. It has long been known (see, e.g., \cite{HaZe85} and the discussion in \cite{PSSZ}) that results obtained for the exponential distribution apply to most continuous non-negative distributions with a positive density at~0. In particular, as discussed in Section~\ref{sec:other}, all our results apply immediately also to the \emph{uniform\/} distribution on $[0,1]$ (except that exponential tail bounds are not as tight). We go here one step further and show that our results also apply when edge weights are \emph{powers\/} of exponential random variables, i.e., when they are of the form $\EXP(1)^s$, where $0<s\le 1$. (We conjecture that the same is true also when $s>1$.)
Thus, our probabilistic model is not as specialized as may seem at first sight.

\subsection{Related results}
\emph{Bidirectional\/} algorithms, which perform a forward search from a source vertex and a backward search from a target vertex, can be used to efficiently find a shortest path between a given pair of vertices (see, e.g., Nicholson \cite{Nicholson66} and Pohl \cite{Pohl71}).
Luby and Ragde \cite{LuRa89} used such a bidirectional algorithm to show that a shortest path from a given source to a given target in
$\G_n(\EXP(1))$ can be found in $O(\sqrt{n}\log n)$ expected time. (They again assume, of course, that the adjacency lists are given in sorted order.) However, the bidirectional search technique does not seem to be applicable for the SSSP problem where distances to all vertices are sought. (Where do we start the backward search from?) Our new \emph{forward-backward\/} algorithm is \emph{not\/} bidirectional. It is a Dijkstra-like unidirectional algorithm that uses some backward scans.

Meyer \cite{Meyer03}, Hagerup \cite{Hagerup06} and Goldberg \cite{Goldberg08} obtained SSSP algorithms with an expected running time of $O(m)$.
The $m$-edge directed input graph may be arbitrary but its edge weights are assumed to be chosen at random from a common non-negative probability distribution. When the edge weights are independent, the running time of these algorithms is $O(m)$ with high probability. Our result differs considerably from these results. Our algorithm runs in $O(n)$ time, which is $o(m)$, on a complete graph with $m=\Omega(n^2)$ edges.

\subsection{Organization of paper}

The rest of this paper is organized as follows. In the next section we briefly review the classical algorithms of Dijkstra and Spira which form the basis of our new algorithms. We also discuss several related algorithms. Before presenting our improved algorithm for finding shortest paths, we present in Section~\ref{S-verify} an improved algorithm for \emph{verifying\/} that a given tree is indeed a tree of shortest paths. This helps us explain the ideas behind the improved algorithm for finding shortest paths in the simplest possible setting. In Section~\ref{S-SSSP} we then present our improved forward-backward shortest paths algorithm. The probabilistic analysis of the new algorithm is given in Section~\ref{sec:probabilistic-analysis} when the edge costs are exponential random variables, and in Section~\ref{sec:other} for the uniform and other distributions which have positive finite density at $0$.
In Section~\ref{sec:Weibull} we extend the analysis to the case in which edge weights are powers of exponential variables, also known as \emph{Weibull\/} random variables.
For the all-pairs shortest paths problem, we show in Section~\ref{S-all-pair} that the sorting of edge costs to build the sorted adjacency lists can be done in $\Theta(n^2)$ time, except with probability exponentially small in $\Theta(n^2)$.
In Section~\ref{S-bucket} we describe an efficient bucket-based implementation of the priority queues used by our new algorithm.
In Section~\ref{S-lower} we give the lower bound for forward-only algorithms.
We end in Section~\ref{S-concl} with some concluding remarks and open problems.

\section{The algorithms of Dijkstra and Spira}\label{S-Spira}


\newcommand{\SPIRA}{
\parbox[t]{3in}{

\SetAlgoFuncName{Algorithm}{anautorefname}
\begin{function}[H]
\DontPrintSemicolon
\SetAlgoRefName{}

\BlankLine
$S\gets \{s\}$ ; $P\gets \varnothing$ \;

\BlankLine
\ForEach{$v\in V$}
{
    $d[v]\gets \infty$ \;
    $p[v]\gets \bot$ \;
    $\RESET(Out[v])$ \;
}

\BlankLine
$d[s]\gets 0$ \;
$\forward(s)$ \;

\BlankLine
\While{$S\ne V$ \AND\ $P\ne\varnothing$} {

    \BlankLine
    $(u,v) \gets \extmin(P)$ \;
    $\forward(u)$ \;

    \BlankLine
    \If{$v\notin S$}
    {
        \BlankLine
        \CommentSty{// \rm New distance found}\;
        \BlankLine

        $d[v] \gets d[u]+c[u,v]$ \;
        $p[v] \gets u$ \;
        $S \gets S\cup \{v\}$ \;
        \BlankLine
        $\forward(v)$ \;
    }

}

\caption{\Spira(\mbox{$G=(V,Out,c),s$})}
\end{function}
}}


\newcommand{\FORWARDSPIRA}{
\parbox[t]{3in}{

\begin{function}[H]

\DontPrintSemicolon
\SetAlgoRefName{}

\BlankLine
$v\gets \NEXT(Out[u])$ \;

\If{$v\ne \bot$}
{
    $\INSERT(P,(u,v),d[u]+c[u,v])$ \;
}

\caption{\forward{}($u$)}
\end{function}
}}

In this section we review the classical SSSP algorithms of Dijkstra and Spira and set the stage for the description of our new SSSP algorithm. We also mention some related algorithms.

\subsection{Dijkstra's algorithm}\label{sub:Dijkstra}

We start with a brief review Dijkstra's algorithm \cite{Di59} for finding a tree of shortest paths from a given source vertex~$s$ in a directed graph $G=(V,E)$ with a non-negative weight (or cost) function $c:E\to\reals^+$ defined on its edges.

Dijkstra's algorithm maintains for each vertex~$v$ a \emph{tentative\/} distance $d[v]$, which is the length of the shortest path from~$s$ to~$v$ discovered so far. Initially $d[s]=0$ while $d[v]=\infty$ for every $v\ne s$. It also maintains a set $S\subseteq V$ which contains vertices whose distance from~$s$ was already found. Initially $S=\varnothing$. Finally, it also maintains a \emph{priority queue\/} $P$ that holds all vertices in $V\setminus S$ whose tentative distance is finite. The key of each vertex in~$P$ is its tentative distance. The algorithm starts by inserting $s$ into~$P$.

In each iteration, Dijkstra's algorithm removes from~$P$ a vertex~$u$ with a smallest tentative distance. The tentative distance~$d[u]$ of~$u$ is then guaranteed to be the distance from~$s$ to~$u$ in the graph, so $u$ is added to~$S$. In addition to that, all \emph{outgoing\/} edges of~$u$ are \emph{relaxed}, i.e., for each outgoing edge $(u,v)\in E$, the algorithm checks whether $d[u]+c[u,v]<d[v]$. If so, then a shorter path to~$v$ was found and the tentative distance of~$v$ is changed to $d[u]+c[u,v]$. If~$v$ is not already in~$P$, it is inserted into~$P$, with key $d[v]$. If $v$ is already in~$P$, then its key is decreased to~$d[v]$.

Dijkstra's algorithm examines each edge of the graph at most once. It inserts at most~$n$ vertices into the priority queue~$P$, and performs at most~$m$ \emph{decrease-key\/} operations and at most~$n$ \emph{extract-min\/} operations, where $n=|V|$ and $m=|E|$. With suitable data structures the running time of Dijkstra's algorithm is $O(m+n\log n)$ time.

\subsection{Spira's algorithm}\label{sub:Spira}

Spira's algorithm \cite{Spira73} attempts to improve on Dijkstra's
algorithm when the outgoing edges of each vertex~$u$ in the graph are
given in \emph{non-decreasing\/} order of weight. In such a setting, a
tree of shortest paths may potentially be found without scanning all
the edges of the graph.

When Dijkstra's algorithm finds the distance to a vertex~$u$, i.e.,
when~$u$ is removed from the priority queue~$P$, it immediately scans
(and relaxes) all its outgoing edges. Spira's algorithm adopts a
lazier approach.  It scans the outgoing edges of~$u$ one by one.  The
algorithm scans an outgoing edge $(u,v)$ only after it finds all
vertices whose distance from~$s$ is smaller than $d[u]+c[u,v']$, where
$(u,v')$ is the edge preceding $(u,v)$ in the adjacency list of~$u$.
To achieve that, the priority queue~$P$ used by Spira's algorithm
holds \emph{edges\/} rather than vertices.  The key of an edge $(u,v)$
in~$P$ is $d[u]+c[u,v]$. Also, if $(u,v)$ is in~$P$ then $d[u]$ is
already set to the correct distance from~$s$ to~$u$.

Spira's algorithm again maintains a set~$S\subseteq V$ that contains all vertices whose distance from~$s$ was already determined. Initially $S=\{s\}$. If $v\in S$, then $d[v]$ is the distance from~$s$ to~$v$. If $v\ne S$, then $d[v]=\infty$. If $v\in S\setminus\{s\}$, then $(p[v],v)$ is the last edge on a path of length~$d[v]$ from~$s$ to~$v$. Initially $d[s]=0$ while $d[v]=\infty$ for every $v\ne s$, and $p[v]=\bot$, for every $v\in V$.

Spira's algorithm starts by scanning the \emph{first\/} outgoing edge $(s,v)$ of~$s$ and inserting it into the priority queue~$P$ with key $d[s]+c[s,v]=c[s,v]$. In each iteration the algorithm extracts an edge $(u,v)$ with the smallest key $d[u]+c[u,v]$ from~$P$. If $(u,v)$ is not the last outgoing edge of~$u$, then the edge $(u,v')$ that follows it in the adjacency list of~$u$ is inserted into~$P$ with key $d[u]+c[u,v']$.  Now, if $v\notin S$, then $d[u]+c[u,v]$, the key of $(u,v)$, is guaranteed to be the distance to~$v$. Thus, $d[v]$ is set to $d[u]+c[u,v]$, $p[v]$ is set to~$u$, $v$ is added to~$S$, and the first outgoing edge $(v,w)$ of~$v$, if there is one, is scanned and inserted into~$P$.

Note that Spira's algorithm inserts an edge $(u,v)$ into~$P$ even if $v\in S$ or if $P$ already contains and edge $(u',v)$ with $d[u']+c[u',v]<d[u]+c[u,v]$.  When $(u,v)$ is extracted from~$P$, the algorithm knows that it is time to scan the next outgoing edge of~$u$.

\begin{figure}[t]
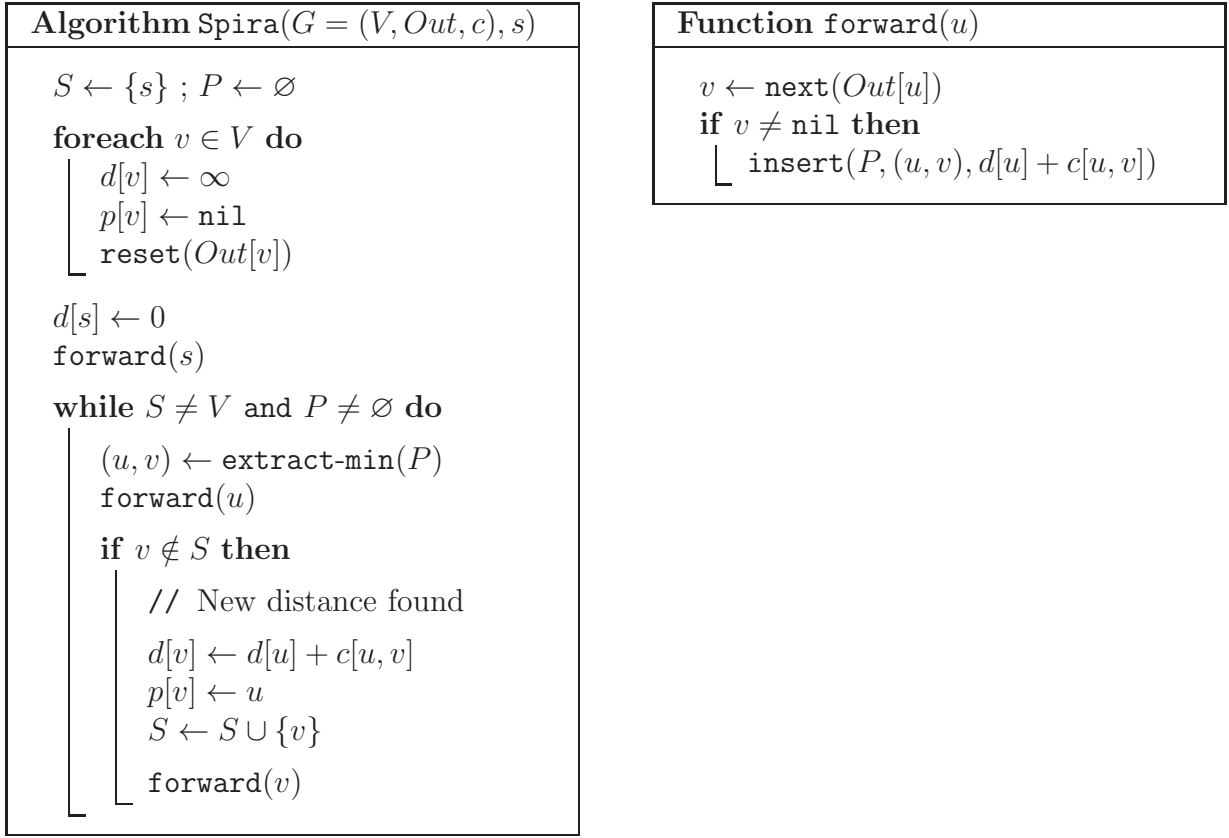

\begin{center}
\SPIRA
\hspace*{20pt}
\FORWARDSPIRA
\end{center}
\vspace*{-10pt} \caption{Spira's algorithm. $\forward(u)$ scans the next outgoing edge of~$u$.} \label{F-Spira}
\end{figure}

Pseudo-code of Spira's algorithm is given in Figure~\ref{F-Spira}. We
discuss it in some detail as most of this code is reused by our
improved algorithm.  Each vertex $u\in V$ has an adjacency list
$Out[u]$ of its outgoing edges, sorted in non-decreasing order of
weight. Although we view $Out[u]$ as a list of edges, each element in
$Out[u]$ is a vertex, the other endpoint of the edge that leaves~$u$.
Each list~$Out[u]$ has a pointer used to sequentially access its
edges. $\RESET(Out[u])$ makes this pointer point to the first edge of
the list, if the list is non-empty.  $\NEXT(Out[u])$ returns the edge
currently pointed to and advances the pointer to the next edge in the
list, or past the end of the list.  If the list is empty, or the
pointer is past the end of the list, then $\NEXT(Out[u])$ returns
$\bot$.

The implementation of Spira's algorithm uses a function $\forward(u)$
that finds the next outgoing edge of~$u$ and inserts it, if it
exists, into~$P$ with key $d[u]+c[u,v]$.  The next outgoing edge is
found by calling $\NEXT(Out[u])$.

Spira \cite{Spira73} analyzed his algorithm in the \emph{end-point independent\/} model mentioned in Section~\ref{sub:average}. Note that $\G_n(\EXP(1))$ is clearly an end-point independent model.

\begin{theorem}[\cite{Spira73}]\label{T-Spira}
Spira's algorithm correctly computes a tree of shortest paths from~$s$
in the input graph $G=(V,E)$.  The expected number of edges scanned by
the algorithm, when edge weights are generated using an end-point
independent process, is at most $(1+o(1))n\log n$.
\end{theorem}

\begin{proof}
  The correctness proof is a simple modification of the correctness
  proof of Dijkstra's algorithm: Whenever an edge $(u,v)$ is extracted
  from the priority queue, unless $v\in S$ already, there is no
  cheaper path to $v$.

  We next bound the expected number of edges
  examined by the algorithm when the edge weights are generated by an end-point independent process.
  We say that
  the algorithm is in stage~$k$ when $|S|=k$. When an edge $(u,v)$ is
  extracted from~$P$ in stage~$k$, the probability that $v\notin S$
  is at least $(n-k)/n$. (More precisely, if $(u,v)$ is the $i$-th
  outgoing edge of~$u$, then this probability is $(n-k)/(n-i)$.) Thus,
  the expected number of edges extracted in stage~$k$ is at most
  $n/(n-k)$, and the expected number of edges extracted during all
  stages is at most $\sum_{k=1}^{n-1} n/(n-k) = (1+o(1)) n\log n$.
  Since the priority queue never has more than~$n$ edges in it,
  the expected number of edge insertions is also at most $(1+o(1)) n \log n$.
\end{proof}

It is not difficult to show, using a slightly more careful analysis, that the expected number of edges scanned by Spira's algorithm, when run on $\G_n(\EXP(1))$, is $(1+o(1))n\log n$.  A $(1+o(1)) n\log n$ lower bound also follows from Theorem~\ref{thm:lower-bound} below.

\subsection{Related algorithms}

The expected number of edges scanned by Spira's algorithm in the end-point independent model is $O(n\log n)$. However, each scanned edge is inserted into the priority queue $P$, and essentially all scanned edges are eventually extracted from~$P$.
Is it possible to lower the number of priority queue operations
without substantially increasing the number of edges scanned?
This is especially interesting if a comparison-based data structure is used to implement~$P$, as the time to extract an edge with the smallest key from~$P$ is then $O(\log n)$ and the total running time of the algorithm becomes $O(n\log^2 n)$.

Dantzig \cite{Dantzig60}, preceding Spira \cite{Spira73} by 13 years, suggested an algorithm that is similar to Spira's algorithm, with one crucial difference. When an edge $(u,v)$ is extracted from~$P$, Spira's algorithm scans the next outgoing edge~$(u,v')$ of~$u$ and inserts it into~$P$ even if $v'\in S$ or $P$ already contains an edge $(u',v')$ with $d[u']+c[u',v']\le d[u]+c[u,v']$. Dantzig's algorithm, on the other hand, starts relaxing the next outgoing edges
$(u,v_1),(u,v_2),\ldots,(u,v_k)$ of~$u$ until it either finds an \emph{improving\/} edge $(u,v_k)$, i.e., an edge satisfying $d[u]+c[u,v_k]<d[v_k]$, or until it exhausts all outgoing edges of~$u$. If an improving edge $(u,v_k)$ is found, then $d[v_k]$ is set to $d[u]+c[u,v_k]$, and only this edge $(u,v_k)$ is inserted into~$P$. Dantzig's algorithm reduces the number of priority queue operations, but greatly increases the number of edges scanned.  The expected number of edges scanned and relaxed by Dantzig's algorithm, when applied to $\G_n(\EXP(1))$ with sorted outgoing adjacency lists, is $\Theta(n^2)$.

Bloniarz \cite{Bloniarz83} obtained a variant of Spira's algorithm that, in the end-point independent model, scans an expected number of $O(n\log n)$ edges and performs an expected number of only $O(n\log^*n)$ priority queue operations. Bloniarz's algorithm uses Dantzig's approach in a very controlled manner. When an edge $(u,v)$ is extracted from~$P$, the algorithm relaxes the next outgoing edges of $u$ until either an improving edge $(u,v')$, i.e., $d[u]+c[u,v']<d[v']$, is found, or until $\log n$ outgoing edges are scanned. If no improving edge is found after scanning $\log n$ edges, the last edge scanned is inserted into~$P$ even if it is not improving.

Moffat and Takaoka \cite{MoTa87} and Mehlhorn and Priebe \cite{MePr97} (see also Takaoka and Hashim \cite{MR2783521,Takaoka12}) obtained further refinements of Bloniarz's algorithm. Their algorithms perform, on average, only $O(n)$ extract-min operations, while still scanning only $O(n\log n)$ edges,
yielding comparison-based algorithms with an expected running time of $O(n\log n)$.

\section{Verifying shortest paths trees}\label{S-verify}

Before considering the problem of \emph{finding\/} a shortest paths tree (SPT), let us consider the easier problem of \emph{verifying\/} that a given tree is indeed a SPT. This will allow us to present the intuition behind the new shortest paths algorithm and explain why the backward scanning of edges is essential for its improved efficiency.

A verification algorithm is an algorithm that receives a weighted directed graph $G=(V,E,c)$, where $c:E\to\reals^+$, and a directed spanning tree~$T$ of~$G$ rooted at a source vertex~$s$. The algorithm should check whether~$T$ is a SPT of~$G$ with source~$s$. The verification algorithm is not allowed to err.

We assume that each vertex~$u$ has an adjacency list $Out[u]$ containing the outgoing edges of~$u$ and an adjacency list $In[u]$ containing the incoming edges of~$u$. Furthermore, we assume that the edges appear in these adjacency lists in non-decreasing order of weight. The tree~$T$ is specified using an array~$p$ of parent pointers. If $s$ is the root of the tree, then $p[s]=\bot$. If $u\ne s$, then $p[u]\ne\bot$ and $(p[u],u)$ is the last edge in the path from $s$ to $u$ within~$T$.

Given a tree~$T$ with root~$s$, let $d[u]$ denote the length of the
path from $s$ to $u$ within the tree.  We have $d[s]=0$, and
$d[u]=d[p[u]]+c[p[u],u]$, for every $u\ne s$, where $c[u,v]$ is the
weight of an edge $(u,v)$.  These formulas lead to an $O(n)$-time
recursive procedure for computing the array $d$ from the array $p$.  By
capping the recursion at depth $n$, we can detect any cycles that might
exist in the graph defined by $p$, and thereby verify that the
array~$p$ specifies a valid tree.

A tree~$T$ is a SPT if and only if $d[u]+c[u,v]\ge d[v]$, or equivalently $c[u,v]\ge d[v]-d[u]$, for every $(u,v)\in E$.

\subsection{A forward-only verification algorithm}

Let $D=\max\{d[u] : u\in V\}$ be the maximal distance in~$T$. The most obvious verification algorithm simply scans the outgoing adjacency list of each vertex~$u$, verifying the condition $c[u,v]\ge d[v]- d[u]$, until it either exhausts the adjacency list of~$u$, or encounters an edge $(u,v)$ for which $c[u,v]\ge D-d[u]$. If $(u,v')$ appears after $(u,v)$ in $Out[u]$, then $c[u,v']\ge c[u,v]\ge D-d[u]\ge d[v']-d[u]$, so $(u,v')$ satisfies the required condition. We refer to this algorithm as the \emph{forward-only\/} verification algorithm.

It is not difficult to verify that the edges examined by this forward-only verification algorithm, when the given tree~$T$ is indeed a tree of shortest paths, are exactly the edges that Spira's algorithm inserts into its priority queue, though not necessarily in the same order. As an immediate corollary of the discussion following Theorem~\ref{T-Spira}, we thus get:

\begin{theorem}\label{T-forward-verify}
The expected number of edges examined by the above forward-only verification algorithm,
when run on a SPT of $\G_n(\EXP(1))$, is $(1+o(1)) n\log n$.
\end{theorem}

In Theorem~\ref{thm:lower-bound} we show that any verification algorithm that only uses the outgoing adjacency lists must inspect an expected number of at least $(1+o(1)) n\log n$ edges when applied to $\G_n(\EXP(1))$. A similar result, for a different randomly weighted graph, was obtained by Mehlhorn and Priebe \cite{MePr97}.

\subsection{A forward-backward verification algorithm}

We next show that by using the incoming adjacency lists as well as the outgoing adjacency lists we can obtain a verification algorithm that runs in $O(n)$ time, with high probability, when given a SPT of $\G_n(\EXP(1))$.
The forward-backward verification algorithm is based on the notion of \emph{pertinent\/} edges.

\begin{definition}[Pertinent edges]
An edge $(u,v)\in E$ is said to be \emph{out-pertinent}, with respect to a given source vertex~$s$ and threshold~$M$, if and only if
\[c[u,v]\leq 2\,(M-d[u])\,.\]
Edge
 $(u,v)$ is said to be \emph{in-pertinent\/} if and only if
\[c[u,v]< 2\,(d[v]-M)\,.\]
(Note that the first inequality is $\leq$ while the second is $<$.)
An edge is said to be \emph{pertinent\/} if it is either out-pertinent or in-pertinent.
We let $\Eoutper$ denote the set of out-pertinent edges, $\Einper$ denote the set of in-pertinent edges, and $\Eper=\Eoutper\cup \Einper$ denote the set of
pertinent edges.
\end{definition}

\begin{remark}
  For any weighted graph and any $M$, every edge in the shortest path tree is either out-pertinent or in-pertinent, and no edge is both.
\end{remark}

The forward-backward verification algorithm sets the threshold~$M$ to be the median distance of the vertices from the source, and then
checks the condition $c[u,v]\ge d[v]-d[u]$ for all pertinent edges, \emph{ignoring\/} all other edges.
The median distance~$M$ can be computed in linear time \cite{find,MR0329916}.
For every vertex~$u$ for which $d[u]\le M$, the
algorithm then checks all outgoing edges of~$u$ of weight at most $2\,(M-d[u])$. This is the \emph{forward\/} scan. For every vertex~$v$ for which $d[v]\ge M$, it then checks all incoming edges of~$v$ of weight less than $2\,(d[v]-M)$. This is the \emph{backward\/} scan.
If all conditions are satisfied, the verification algorithm accepts~$T$. The correctness of the algorithm follows from the following lemma.

\begin{lemma}\label{L-verify}
  If $c[u,v]\ge d[v]-d[u]$ for every pertinent edge $(u,v)$, then $c[u,v]\ge d[v]-d[u]$ for every $(u,v)\in E$.
\end{lemma}

\begin{proof}
If $(u,v)\notin \Eper$, then
$c[u,v] > 2\,(M-d[u])$ and $c[u,v] \ge 2\,(M-d[v])$. Thus,
\[ c[u,v] \;>\; \frac12 \biggl( 2\,(M-d[u]) + 2\,(d[v]-M) \biggr) \;=\; d[v]-d[u]\,,\]
as required.
\end{proof}

The verification algorithm is correct with any choice of~$M$.  Letting~$M$ be the median distance minimizes the running time in many interesting cases.  For $\G_n(\EXP(1))$ we show in Section~\ref{sec:probabilistic-analysis} that the number of pertinent edges with respect to the median distance is $\Theta(n)$ both in expectation and with high probability.

As a consequence, we get:
\begin{theorem}\label{T-verify}
  The running time of the forward-backward verification algorithm,
  when run on $\G_n(\EXP(1))$ with sorted adjacency lists, is
  $\Theta(n)$, with very high probability.  The probability that the
  running time exceeds $\Theta(n+\Delta)$ decays exponentially in $\Delta$.
\end{theorem}
\begin{proof}
  This follows from Theorem~\ref{T-pertinent-exp} which we prove in Section~\ref{sec:probabilistic-analysis}.
\end{proof}

\begin{remark}
  For the purposes of verifying a SPT, it suffices to check the edges
  $(u,v)$ for which $(u,v)$ is in the tree, or $c[u,v]<2\,(M-d[u])$, or
  $c[u,v]<2\,(d[v]-M)$.  However, for the purposes of finding the SPT,
  our algorithm will also look at edges for which $c[u,v]=2\,(M-d[u])$,
  and up to $n$ additional edges.
\end{remark}


\newcommand{\SSSP}{
\parbox{3.0in}{
\SetAlgoFuncName{Algorithm}{anautorefname}
\begin{function*}[H]
\DontPrintSemicolon
\SetAlgoRefName{}

\BlankLine
$S\gets \{s\}$ ; $M\gets \infty$ ; $n\gets|V|$ \;
$P\gets \varnothing$ ; $Q\gets \varnothing$ \;

\BlankLine
\ForEach{$v\in V$}
{
    \BlankLine
    $d[v]\gets \infty$ ; $p[v]\gets \bot$ \;
    $out[v]\gets \TRUE$ ;
    $Req[v] \gets \varnothing$ \;
    \BlankLine
    $\RESET(Out[v])$ ;
    $\RESET(Req[v])$ \;
    $\RESET(In[v])$ \;
}

\BlankLine
$d[s]\gets 0$ \;
$\forward(s)$ \;

\BlankLine
\While{$S\ne V$ \AND\ $P\ne\varnothing$} {

    \BlankLine
    $(u,v) \gets \extmin(P)$ \;
    $\forward(u)$ \;

    \BlankLine
    \If{$v\notin S$}
    {
        \BlankLine
        \CommentSty{// \rm New distance found}\;
        \BlankLine
        $d[v] \gets d[u] + c[u,v]$ \;
        $p[v] \gets u$ \;
        $S \gets S\cup \{v\}$ \;
        \BlankLine
        $\forward(v)$ \;

        \BlankLine
        \If{$|S| = \lceil n/2 \rceil$}
        {
            \BlankLine
            $M\gets d[v]$ \;
            \BlankLine
            \ForEach{$w\notin S$}
            {
                $\backward(w)$
            }
        }

    }

        \BlankLine
        \CommentSty{// \rm Find more \emph{in-pertinent\/} edges}\;
        \BlankLine

        \While{$Q\neq\varnothing$ \AND\ $\ \ \ \minkey(Q)< 2\,(\minkey(P){-}M)$}
        {
            \BlankLine
            $(u,v) \gets \extmin(Q)$ \;
            \BlankLine
            \If{$v\notin S$}
            {
                $\backward(v)$ \;
                $\request(u,v)$ \;
            }
        }

}
\caption{\sssp(\mbox{$G=(V,In,Out,c),s$})}
\end{function*}
}}


\newcommand{\REQUEST}{
\parbox{3in}{
\begin{function}[H]
\DontPrintSemicolon
\BlankLine
    $\append(Req[u],v)$ \;
    \If{$u\in S$ \AND\ \NOT\ $active[u]$}{\CommentSty{// \rm urgent request}\;$\forward(u)$}
\caption{\request(\mbox{$u,v$})}
\end{function}
}}


\newcommand{\REQUESTIF}{
\parbox{3in}{
\begin{function}[H]
\DontPrintSemicolon
\BlankLine
\If{$c[u,v] > 2\,(M{-}d[u])$}{
    $\append(Req[u],v)$ \;
    \If{$u\in S$ \AND\ \NOT\ $active[u]$}{\CommentSty{// \rm urgent request}\;$\forward(u)$}
}
\caption{\request{}{}(\mbox{$u,v$})}
\end{function}
}}


\newcommand{\FORWARD}{
\parbox[t]{3in}{

\begin{function}[H]

\DontPrintSemicolon
\SetAlgoRefName{}

\BlankLine
\If{$out[u]$}
{
    \BlankLine
    \CommentSty{// \rm find next \emph{out-pertinent\/} edge}\;
    \BlankLine
    $v\gets \NEXT(Out[u])$ \;
    \If{$v= \bot$ \OR\ $c[u,v]> 2\,(M{-}d[u])$}
    { $out[u]\gets \FALSE$ }
}

\BlankLine
\If{\NOT\ $out[u]$}
{
    \BlankLine
    \CommentSty{// \rm find next \emph{in-pertinent\/} edge}\;
    \BlankLine
    $v\gets \NEXT(Req[u])$ \;
}
\BlankLine
$active[u] \gets v\ne \bot$ \;
\BlankLine
\If{$active[u]$}
{
    $\INSERT(P,(u,v),d[u]+c[u,v])$ \;
}

\caption{\forward($u$)}
\end{function}
}}


\newcommand{\BACKWARD}{
\parbox[t]{3in}{

\begin{function}[H]

\DontPrintSemicolon

\BlankLine
$u\gets \NEXT(In[v])$ \;
\If{$u\ne\bot$}
{
    $\INSERT(Q,(u,v),c[u,v])$ \;
}

\caption{\backward(\mbox{$v$})}
\end{function}
}}


\newcommand{\NEXTOUT}{
\parbox[t]{3in}{

\begin{function}[H]

\DontPrintSemicolon

\BlankLine
$v\gets \current(Out[u])$ \;
\If{$v\ne \bot$ \AND\ $c[u,v]<2\,(M-d[u])$}
{
    $\ADVANCE(Out[u])$ \;
    \Return{$v$}
}

\BlankLine
$v'\gets \current(Req[u])$ \;
\lIf{$v'\ne \bot$}{$\ADVANCE(Req[u])$}
\Return{$v'$}

\caption{\nextout(\mbox{$u$})}
\end{function}
}}


\SetAlgoFuncName{Function}{anautorefname}

\begin{figure*}[t]
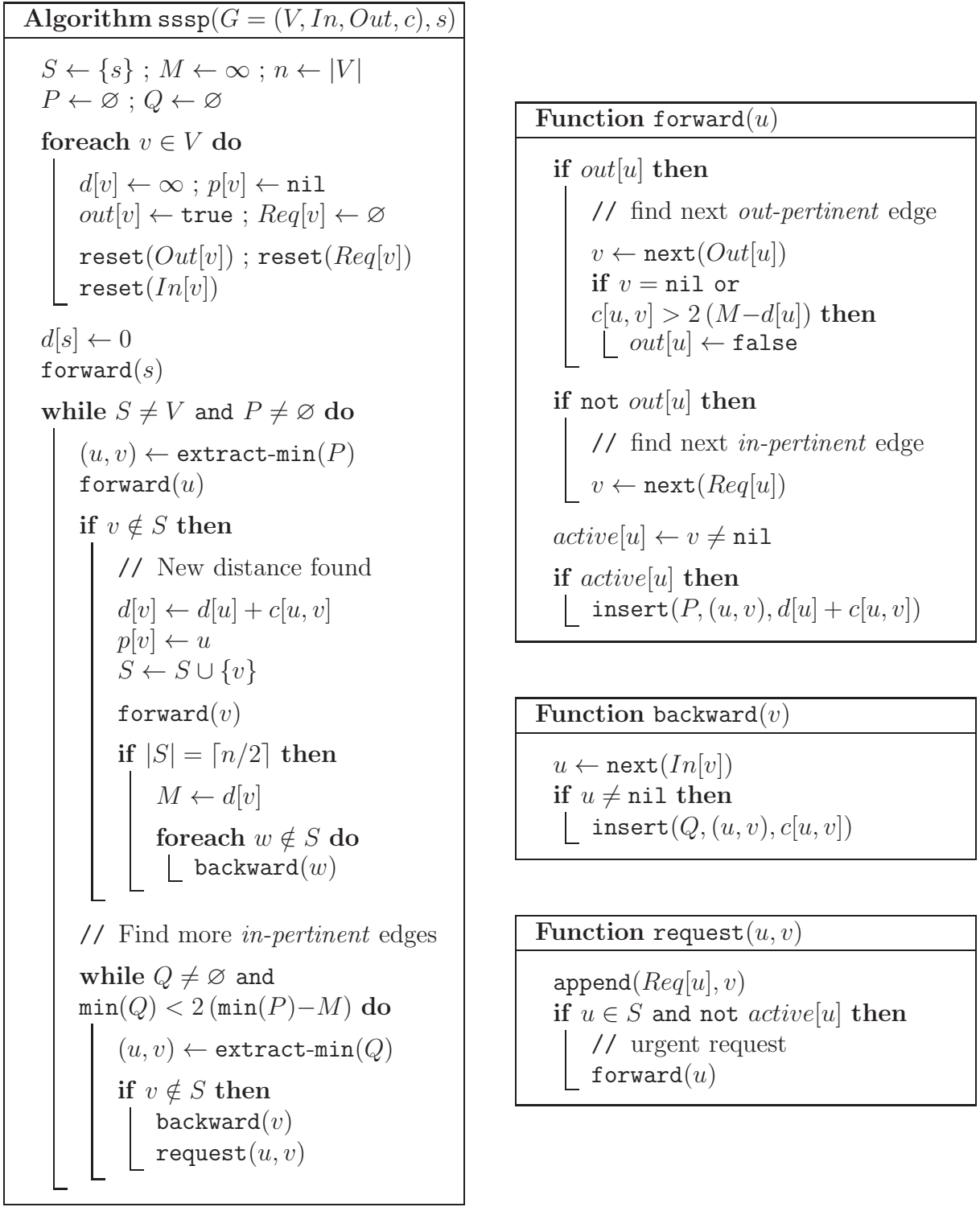

\begin{center}
\SSSP
\hspace*{20pt}
\parbox{3in}{
\FORWARD\\[19pt]
\BACKWARD\\[19pt]
\REQUEST
}
\end{center}
\vspace*{-10pt} \caption{The new forward-backward SSSP algorithm.  The inner \textbf{while} loop is only executed when $|S|\geq\lceil n/2\rceil$.  \forward($u$) scans the next outgoing edge of~$u$, and \backward($v$) scans the next incoming edge of~$v$.
\request($u,v$) requests a forward scan of edge $(u,v)$.
} \label{F-sssp}
\end{figure*}

\section{The forward-backward shortest paths algorithm}\label{S-SSSP}

Our goal in this section is to develop a single-source shortest paths
algorithm that matches the performance of the forward-backward
verification algorithm of the previous section.  To achieve that, a
majority of the edges examined by the algorithm must be
\emph{pertinent}.  Recall that an edge $(u,v)$ is out-pertinent if
$c[u,v]\le 2\,(M-d[u])$, and in-pertinent if $c[u,v]<
2\,(d[v]-M)$. The obvious difficulty is that this definition uses the
distances $d[u]$ and $d[v]$, and the median distance $M$, which are
the quantities we want to compute. Furthermore, we must rule out
almost all non-pertinent edges without even looking at them.

The new algorithm is composed of two stages. In the first stage the algorithm finds distances to the closest $\lceil n/2 \rceil$ vertices, and hence also the median distance $M$. In its first stage, the algorithm behaves exactly like Spira's algorithm described in Section~\ref{sub:Spira}. It is not difficult to check that all edges extracted from priority queue~$P$ during the first stage are out-pertinent. The at
most $\lceil n/2 \rceil$ edges residing in~$P$ at the end of first iteration may not be pertinent. All edges inserted into~$P$ in the second stage would again be pertinent.

The second stage of the algorithm finds the distances to the remaining $\lfloor n/2\rfloor$ vertices. (We assume here, for simplicity, that all vertices in the graph are reachable from~$s$.)
In its second stage, the algorithm mimics the behavior of Spira's algorithm on the graph obtained by removing all non-pertinent edges. The algorithm is correct as we know that non-pertinent edges do not participate in shortest paths. The challenge, of course, is to do it while examining only $O(n)$ non-pertinent edges. Out-pertinent edges are relatively easy to identify, as they appear in the beginning of the sorted outgoing adjacency lists of the vertices. To quickly identify in-pertinent edges, we also need the sorted incoming adjacency lists.

During the second stage, we already know~$M$, the median distance.
An edge $(u,v)$ is scanned forward only after the distance~$d[u]$ to~$u$ is found. Thus, we can easily check whether $(u,v)$ is out-pertinent. If it is not, then all edges following $(u,v)$ in $Out[u]$, the outgoing adjacency list of~$u$, are also not out-pertinent. We thus stop scanning $Out[u]$. However, $u$ may still have outgoing edges that are in-pertinent, and it is essential that we scan them, in the correct order.

Suppose that $(u',v')$ is the last edge extracted so far from~$P$. Let $K=key[u',v']=d[u']+c[u',v']$. As distances are found in non-decreasing order, for every $v\notin S$ we have $d[v]\ge K$. Thus, if $v\notin S$, then all incoming edges of~$v$ of weight less than $2\,(K-M)$ are in-pertinent. As $In[v]$, the incoming adjacency list of~$v$, is sorted, we can easily access these edges. When a new in-pertinent edge $(u,v)$ is discovered, we cannot immediately scan in the forward direction as~$u$ may have cheaper outgoing pertinent edges that were not scanned yet, or more seriously,
a shortest path to~$u$ may not have been found yet, so outgoing edges of~$u$ cannot be scanned yet.
When a new in-pertinent edge $(u,v)$ is discovered, we simply \emph{request\/}~$u$ to forward scan it at the appropriate time in the future. For every vertex~$u$ we maintain a list~$Req[u]$ that contains the requested outgoing edges of~$u$. When a new in-pertinent edge $(u,v)$ is found, it is appended to~$Req[u]$. Edges in $Req[u]$ again appear in non-decreasing order of weight. If~$u$ had exhausted all its outgoing pertinent edges when a new in-pertinent edge~$(u,v)$ is discovered, the request~$(u,v)$ is considered to be \emph{urgent}, and~$(u,v)$ is immediately scanned and added to~$P$.

To make sure that in-pertinent edges are discovered in non-decreasing order of weight, and that the lists $Req[u]$ are sorted, the new algorithm uses a second priority queue~$Q$. Each vertex~$v\notin S$ sends its first incoming edge, which is not yet known to be in-pertinent, to~$Q$. The key of each edge $(u,v)$ in~$Q$ is simply its weight $c[u,v]$. Each iteration of the second stage of the algorithm ends with a loop in which edges of minimum weight are extracted from~$Q$. Each extracted edge is guaranteed to be in-pertinent. When an edge $(u,v)$ is extracted from~$Q$, it is immediately replaced by the edge $(u',v)$ following $(u,v)$ in $In[v]$, if such an edge exists. We refer to this as a \emph{backward scan\/} of the edge~$(u',v)$.

Let $u_k$ be the last vertex added to~$S$, i.e., the vertex with the
largest distance found so far. Let $u_{k+1}$ be the next vertex to be
discovered, i.e., the vertex with the next highest distance from~$s$
(which might be the same distance as $u_k$).  It was mentioned above
that if~$K$ is the key of the last edge extracted from~$P$, then every
edge $(u,v)$ satisfying $v\notin S$ and $c[u,v]<2\,(K-M)$ is
in-pertinent. Note that $d[u_k]\le K\le d[u_{k+1}]$. In fact, the same
claim is true if $K$ is replaced by $d[u_{k+1}]$. The problem, of
course, is that~$u_{k+1}$ and~$d[u_{k+1}]$ are not yet known to the
algorithm. Unfortunately, if we only rely on the condition
$c[u,v]<2\,(K-M)$, the algorithm may fail to realize that the last
edge $(u,u_{k+1})$ on the shortest path to~$u_{k+1}$ is in-pertinent,
and thus fail to correctly identify~$u_{k+1}$.  Instead of using the
condition $c[u,v]<2\,(K-M)$, the algorithm therefore uses the
condition $c[u,v]<2\,(\min(P)-M)$, where $\min(P)$ is the minimal key
in~$P$. Proving that using this condition allows the algorithm to
correctly identify $u_{k+1}$, and that all edges satisfying the
condition are in-pertinent, is one of the subtler points in the
correctness proof of the algorithm.  (Note that $\min(P)$, which may
change as a result of urgent requests, is used here as an ``estimate''
of $d[u_{k+1}]$, even though $\min(P)$ is sometimes larger than
$d[u_{k+1}]$.)

To summarize, the new algorithm mimics the behavior of Spira's algorithm on the subgraph defined by pertinent edges. It does so while examining only $O(n)$ edges of the original graph that are not pertinent. (This is not a probabilistic statement. It holds for any input graph.) The algorithm is composed of two stages. In its first stage, which lasts until distances to the first $\lceil n/2\rceil$ vertices are found, the algorithm behaves exactly like Spira's algorithm. Essentially all edges scanned by the algorithm in the first stage are guaranteed to be out-pertinent. In its second stage, the algorithm explicitly checks the pertinence conditions. It stops scanning outgoing adjacency lists when non out-pertinent edges are reached. To identify in-pertinent edges, the algorithm performs \emph{backward\/} scans from vertices whose distance from~$s$ was not found yet. These backward scans make sure that in-pertinent edges are available in time for forward scans in subsequent iterations of the algorithm.

A full description of the algorithm is given in Section~\ref{SS-detailed}. We prove its correctness in Section~\ref{SS-correctness}, and analyze its performance in Section~\ref{SS-complexity}.

\subsection{Description of the algorithm}\label{SS-detailed}

The input to the algorithm is a weighted directed graph $G=(V,E,c)$, where $c:E\to\reals^+$, and a source $s\in V$.
Each vertex $u\in V$ has a list $Out[u]$ of its outgoing edges and a list $In[u]$ of its incoming edges of~$u$. Both lists are sorted in non-decreasing order of cost. Although we view $Out[u]$ and $In[u]$ as list of edges, each element in them is a vertex, the other endpoint of the edge that leaves or enters~$u$.
Every vertex $u\in V$ also has a second list~$Req[u]$ of outgoing edges. Initially this list is empty. $Req[u]$ contains edges whose scan was specifically \emph{requested}, as explained above. All requested edges are in-pertinent.

Each such $In$, $Out$ or $Req$ adjacency list~$L$ has a \emph{pointer\/} used to sequentially access its edges. $\RESET(L)$ makes this pointer point to the first edge of the list, if the list is non-empty.
$\NEXT(L)$ returns the edge currently pointed to and advances the pointer to the next edge in the list, or past the end of the list. If the list is empty, or the pointer is past the end of the list, then $\NEXT(L)$ returns $\bot$. If an edge is appended to~$L$ when the pointer is past the end of the list, the pointer is set to point to the newly added edge.

For each vertex $u$, the algorithm maintains a bit $out[u]$ which is set if $Out[u]$ may still contain unscanned out-pertinent edges. Initially $out[u]$ is set to $\TRUE$. When the next pertinent outgoing edge of~$u$ is sought, and $out[u]$ is set, the algorithm looks at the next available edge from $Out[u]$. If this edge is out-pertinent, it is used. If it is not out-pertinent, then $out[u]$ is set to $\FALSE$. When $out[u]$ is \FALSE, the next edge from $Req[u]$, if there is one, is used; it is guaranteed to be in-pertinent.

The algorithm maintains a set~$S\subseteq V$ that contains all the vertices $v$ whose distance from~$s$ was already found. Initially $S=\{s\}$. If $v\in S$, then $d[v]$ is the distance from~$s$ to~$v$ in the graph. If $v\in S\setminus\{s\}$, then $(p[v],v)$ is the last edge on a shortest path from~$s$ to~$v$ in the graph. If $v\notin S$, then $d[v]=\infty$ and $p[v]=\bot$. We also have $p[s]=\bot$. The algorithm also maintains the size of the set $S$. (There is actually no need to explicitly maintain $S$. The condition $v\in S$, used below, can be replaced by the condition $d[v]<\infty$, and the condition $S\ne V$ can be replaced by $|S|\ne|V|$. However, it is useful to refer to the set~$S$ by name.)

The algorithm maintains two priority queues $P$ and $Q$. The first priority queue~$P$ is analogous to the priority queue used by Spira's algorithm. The second priority queue~$Q$ is used to identify in-pertinent edges as explained above. At any stage during the operation of the algorithm, each vertex $u$ has at most one outgoing edge~$(u,v)$ in~$P$. Essentially all these edges are pertinent. Similarly, each vertex~$v$ has at most one incoming edge~$(u,v)$ in~$Q$. When $(u,v)$ is inserted into~$Q$ we have $v\notin S$. However, $v$ may be added to~$S$ before~$(u,v)$ is extracted from~$Q$. All edges extracted from~$Q$ are in-pertinent edges.

Pseudo-code of the new forward-backward single-source shortest paths algorithm is given in Figure~\ref{F-sssp}. It starts with straightforward initializations. In particular, $M$ is initialized to~$\infty$. The algorithm uses a function $\forward(u)$ that finds the next pertinent outgoing edge $(u,v)$ of~$u$, if there is one, and inserts it into~$P$. In the first stage of the algorithm all edges are assumed to be pertinent. The algorithm starts by calling $\forward(s)$ to insert the first outgoing edge of~$s$ into~$P$.

$\forward(u)$ works as follows. If $out[u]$ is \TRUE, it uses $\NEXT(Out[u])$ to obtain the next outgoing edge $(u,v)$ from $Out[u]$, if there is one. If~$(u,v)$ exists, it checks whether $c[u,v]\le 2\,(M-d[u])$. If the algorithm is still in its first stage, then $M=\infty$ and the condition is automatically satisfied. If the algorithm is already in its second stage and the condition is satisfied, then $(u,v)$ is out-pertinent. If~$(u,v)$ does not exist, or fails the condition, then $out[u]$ is set to \FALSE, as~$Out[u]$ does not contain additional out-pertinent edges. If~$out[u]$ is \FALSE, $\forward(u)$ uses $\NEXT(Req[u])$ to obtain the next edge~$(u,v)$ from $Req[u]$. If an appropriate edge~$(u,v)$ is found, from either $Out[u]$ or $Req[u]$, then~$u$ is said to be \emph{active}, $active[u]$ is set to \TRUE, and~$(u,v)$ is inserted into~$P$. If no next edge $(u,v)$ is found, then~$u$ is said to be \emph{inactive}, and $active[u]$ is set to \FALSE.

The operation of algorithm is composed of \emph{iterations}.  Each iteration starts by extracting an edge $(u,v)$ of minimum key from~$P$ and by calling $\forward(u)$ to scan the next outgoing edge of~$u$, if any, and add it to~$P$. If $v\notin S$, then, as we shall see, $d[u]+c[u,v]$ is the distance from~$s$ to~$v$, so $d[v]$ is set to $d[u]+c[u,v]$, $p[v]$ is set to~$u$, $v$ is added to~$S$ and its first outgoing edge is scanned by calling $\forward(v)$. This is all that is done in an iteration during the \emph{first stage\/} of the algorithm, i.e., until the $\lceil n/2 \rceil$-th vertex is added to~$S$.  The behavior of the algorithm in the first stage is thus identical to the behavior of Spira's algorithm. If~$v$ is the $\lceil n/2 \rceil$-th vertex is added to~$S$, then $M$ is set to $d[v]$, which is the median distance, and the algorithm enters its second stage which lasts until distances to all vertices reachable from~$s$ are found.

The second stage starts by backward scanning the first incoming edge of each vertex which is not yet in~$S$ and inserting it into the priority queue~$Q$. Backward scans are performed using $\backward(v)$.

Each iteration during the second stage of the algorithm ends with the execution of an inner while loops that identifies new in-pertinent edges. While $Q\ne\varnothing$ and $\min(Q)<2\,(\min(P)-M)$, an edge $(u,v)$ of minimum weight is extracted from~$Q$. If $v\notin S$, then the next incoming edge of~$v$ is scanned and inserted into~$Q$ by calling $\backward(v)$. The forward scan of $(u,v)$, which is guaranteed to be in-pertinent, is then requested by calling $\request(u,v)$. $\min(P)$ and $\min(Q)$ in the condition above are the minimum keys of elements contained in~$P$ and~$Q$, respectively.  (If $P$ is empty, then $\min(P)$ is taken to be $\infty$.) If the request of an edge $(u,v)$ is urgent, then this edge is immediately inserted into~$P$, which may decrease $\min(P)$. The inner while loop is not executed during the first stage, as $Q$ become non-empty only at the end of the first stage.

\begin{wrapfigure}{r}{3in}
\capstart
    \vspace*{-\baselineskip}
    \REQUESTIF
\caption{Alternate \request\ function, with test of non-out-pertinence.}
\label{requestIf}
\end{wrapfigure}
Requesting an edge $(u,v)$ is done by calling $\request(u,v)$. We shall prove that every requested edge is in-pertinent, and therefore not out-pertinent.  To simplify the correctness proof, we assume at first that $\request(u,v)$ explicitly checks that $(u,v)$ is not out-pertinent, as shown at right.  We later prove that this test is unnecessary, as it is always satisfied.  $\request(u,v)$ appends $(u,v)$ to $Req[u]$. If $u\in S$ and $active[u]$ is \FALSE, then request is urgent, and $\forward(u)$ is called immediately to scan $(u,v)$.

\subsection{Correctness of the algorithm}\label{SS-correctness}

We begin with some technical lemmas that play a central role in the correctness proof.  Let~$d_{v}$ be the distance from~$s$ to~$v$ in the input graph.  Our goal is to show that when the algorithm terminates, $d[v]=d_{v}$, for every $v\in V$.  The first (obvious) lemma claims that the algorithm never underestimates distances.

\begin{lemma}\label{L-delta}
  At any stage of the forward-backward algorithm, $d_{v}\le d[v]$, for every $v\in V$.
\end{lemma}

\begin{proof}
  The claim obviously holds when $d[s]$ is set to~$0$ and all other
  $d[v]$'s set to $\infty$. Tentative distances are then only updated
  by the command $d[v]\gets d[u]+c[u,v]$, where $(u,v)\in E$, which
  clearly preserves the condition.
\end{proof}

Let $key[u,v]=d[u]+c[u,v]$ be the key of an edge $(u,v)$ when it is inserted into~$P$. ($d[u]$ does not change after that moment.)
Parts $(i)$ and $(iv)$ of the second (technical) lemma claim that~$P$ and~$Q$ are \emph{monotone\/} priority queues, i.e., the keys of the successive edges extracted from them are monotonically non-decreasing.

\begin{lemma}\label{L-order}
For the forward-backward algorithm,
\begin{enumerate}[label=(\roman*)]
\item \label{Q-extract} If $(u,v)$ is extracted from~$Q$ before $(u',v')$, then $c[u,v]\le c[u',v']$.
\item \label{Req-insert} If $(u,v)$ is inserted into~$Req[u]$ before $(u,v')$, then $c[u,v]\le c[u,v']$.
\item \label{P-Req} If $(u,v)$ is inserted into~$P$ from $Out[u]$ and $(u,v')$ is requested,
 then $c[u,v]\leq c[u,v']$.
\item \label{P-insert} If $(u,v)$ is inserted into~$P$ before $(u,v')$, then $c[u,v]\le c[u,v']$.
\item \label{P-extract} If $(u,v)$ is extracted from~$P$ before $(u',v')$, then $key[u,v]\le key[u',v']$.
\end{enumerate}
\end{lemma}

\begin{proof} We prove the claims by induction on the number of steps performed by the algorithm.

$\ref{Q-extract}$.  When an edge $(u,v)$ is extracted from~$Q$, at most one edge is added to~$Q$ in its stead. This edge, if it exists, is the edge $(u',v)$ following $(u,v)$ in $In[v]$. As $In[v]$ is sorted, we have $c[u,v]\le c[u',v]$. The next edge $(u'',v'')$ extracted from~$Q$ must therefore satisfy $c[u,v]\le c[u'',v'']$.

$\ref{Req-insert}$.  As edges are appended to $Req[u]$ in the order in which they are extracted from~$Q$, claim~$\ref{Req-insert}$ follows immediately from~$\ref{Q-extract}$.

$\ref{P-Req}$.
Let $(u,v_1),(u,v_2),\ldots$ be the outgoing edges of~$u$ in the order in which they appear in~$Out[u]$. $\forward(u)$ uses edges from $Out[u]$, until it either exhausts $Out[u]$, or until it reaches the first edge $(u,v_{k+1})$ which fails the test $c[u,v_{k+1}]\le 2\,(M-d[u])$.
We claim that all edges $(u,v_1),\ldots,(u,v_{k-1})$ satisfy
 $c[u,v_i]\le 2\,(M-d[u])$, for $1\le i<k$, where $M$ is the median distance.
If $(u,v_k)$ is out-pertinent, the claim is clearly satisfied. If not, then $(u,v_k)$ must have been inserted into~$P$ in the first phase (when $M=\infty$), as otherwise it would have failed the test $c[u,v_k]\le 2\,(M-d[u])$. It follows that $(u,v_1),\ldots,(u,v_{k-1})$ must have been extracted from~$P$ during the first phase.  Every edge $(u,v)$ extracted from $P$ during the first phase satisfies $d[u]+c[u,v]\leq M$, and hence $c[u,v]\leq M-d[u]\le 2\,(M-d[u])$, as claimed.

Edges are inserted into $Req[u]$ only in the second phase of the algorithm when~$M$ is already set.
Every edge $(u,v)$ inserted into $Req[u]$ satisfies $c[u,v]> 2\,(M-d[u])$.  (We assume, for now, that $\request(u,v)$ explicitly checks this condition as in Figure~\ref{requestIf}; we will see in Remark~\ref{L-condition-true} that this condition is always true.)  It follows that every such edge cannot be one of $(u,v_1),\ldots,(u,v_{k-1})$.  As $(u,v_k)$ is cheapest among all other edges of $Out[u]$, it follows that $c[u,v_k]\le c[u,v]$, as required.  This establishes $\ref{P-insert}$. (We note that $(u,v_k)$ may actually be inserted into~$P$ twice, once from~$Out[u]$ and once from~$Req[u]$.  It is, however, the only edge for which this may happen. It is not difficult to modify the code so that this would not happen, but there is no real need for that.)

$\ref{P-insert}$. outgoing edges of~$u$ are inserted into~$P$ by $\forward(u)$, which first uses edges from~$Out[u]$, and then uses edges from $Req[u]$. As $Out[u]$ and $Req[u]$ are both sorted, $\ref{P-insert}$ follows from $\ref{P-Req}$.

$\ref{P-extract}$.
The extraction of an edge $(u,v)$ from~$P$ may cause the insertion of several edges into~$P$.
To prove claim~$\ref{P-extract}$, it suffices to show that for every edge $(u',v')$ inserted into~$P$ after the extraction of $(u,v)$ we have $key[u,v]\le key[u',v']$.
This ensures that the next edge $(u'',v'')$ extracted from~$P$ satisfies $key[u,v]\le key[u'',v'']$.

The first edge inserted into~$P$ after the removal of $(u,v)$ is $(u,v')$, the next outgoing or requested edge of~$u$, if it exists. By $\ref{P-insert}$, we have $c[u,v]\le c[u,v']$ and thus $key[u,v]\le key[u,v']$.

The second edge that might be inserted into~$P$ is $(v,w)$, the first outgoing or requested edge of~$v$, if $(u,v)$ is the first extracted edge entering~$v$. In this case we again have $key[u,v] = d[v] \le key[v,w]$.

Finally, more edges may be added to~$P$ as a result of urgent requests. Suppose that $(u',v')$ is added to~$P$ as a result of such an urgent request.  We begin by showing that $c[u',v']\ge 2\,(key[u,v]-M)$. If the current iteration is the first iteration of the second phase, then $key[u,v]=M$, and the claim is obvious. Suppose instead that the current iteration is not the first iteration of the second phase.
At the end of the previous iteration, the edge of minimum key in~$P$ is $(u,v)$, so each edge $(u'',v'')$ for which $c[u'',v'']< 2\,(key[u,v]-M)$ is requested before the previous iteration ends.  As $(u',v')$ is requested in the current iteration, it was not requested in the previous iteration, so $c[u',v']\geq 2\,(key[u,v]-M)$, as claimed.
Also, $(u',v')$ must be a non out-pertinent edge, thus $c[u',v']>2\,(M-d[u'])$. Combining these two inequalities, as in the proof of Lemma~\ref{L-verify}, we get $c[u',v']>key[u,v]-d[u']$, or equivalently $key[u',v']>key[u,v]$.

We have thus shown that all edges $(u',v')$ inserted into~$P$ after the extraction of $(u,v)$ satisfy $key[u,v]\le key[u',v']$.  This establishes $\ref{P-extract}$.
\end{proof}

\begin{lemma}\label{L-inP}
  If at some time $u\in S$ and $(u,v)$ is an out-pertinent or a requested edge that
  has not been extracted from~$P$, then at that time~$P$ must contain an edge
  $(u,v')$ (possibly $v'=v$) with $c[u,v']\le c[u,v]$.
\end{lemma}

\begin{proof}
When $u$ is added to~$S$, the first edge of $Out[u]$ is added to~$P$. When an edge of $Out[u]$ is extracted from~$P$, it is immediately replaced by the edge following it in $Out[u]$, unless the algorithm is in its second phase and the following edge is not out-pertinent. As the edges in $Out[u]$ appear in sorted order, the claim holds if~$(u,v)$ is out-pertinent.

Suppose now that $(u,v)$ is a requested edge which is not out-pertinent.
If $u$ is inactive at the time $(u,v)$ is requested, the request is urgent and $(u,v)$ is immediately inserted into~$P$.
If edge $(u,v)$ has not been inserted into $P$, then there is some other edge
$(u,v')$ in $P$.
If $(u,v')$ is a requested edge, it was requested earlier than
$(u,v)$, so by Lemma~\ref{L-order}$\ref{Req-insert}$ we have $c[u,v']\leq c[u,v]$.
If $(u,v')$ was not requested, by Lemma~\ref{L-order}$\ref{P-Req}$ we have $c[u,v']\leq c[u,v]$.
\end{proof}

\begin{lemma} \label{L-reachable}
  When the forward-backward algorithm terminates, $S$ is the set of vertices reachable from the source.
\end{lemma}
\begin{proof}
  Since $\forward$ is only invoked on vertices in $S$, by induction,
  priority queue $P$ only contains edges reachable from the source
  $s$, and $S$ only contains vertices reachable from $s$.  If the
  algorithm stops with $|S|<\lceil n/2\rceil$, then it stops with
  $M=\infty$, and since $P=\varnothing$, all outgoing edges of each
  vertex in $S$ have been explored, so $S$ is the set of vertices
  reachable from $s$.  If the algorithm stops with $\lceil
  n/2\rceil\leq|S|<n$, then it stops with $P=\varnothing$, so on the
  previous iteration all incoming edges to vertices not in $S$ were
  requested, but nonetheless not enqueued in $P$.  Any such edge
  $(u,v)$ was not already in $P$, so it is not out-pertinent, and
  since it was not enqueued, $u\notin S$.  In this case, again $S$ is
  the set of vertices reachable from $s$.  The final case is that the
  algorithm stops with $|S|=n$, in which case trivially $S$ consists
  of the vertices reachable from $s$.
\end{proof}

\begin{lemma} \label{L-in-request}
  When an edge $(u,v)$ is extracted from priority queue $P$, all incoming in-pertinent edges of~$v$ have already been requested.
\end{lemma}
\begin{proof}
  If edge $(w,v)$ is in-pertinent, then $d[v]\geq d_{v}>M$, so the
  iteration prior to the extraction of edge $(u,v)$ was in the second phase.
  At the time the previous iteration ended $\min(P)=d[u]+c[u,v]$,
  so at that time
  \[c[w,v]<2\,(d_{v}-M)\leq 2\,((d[u]+c[u,v])-M) = 2\,(\min(P)-M),\] so
  by the end of the previous iteration, edge $(w,v)$ has been requested.
\end{proof}

We are now ready for the proof of the following theorem.

\begin{theorem}\label{T-correct}
  The forward-backward single-source shortest paths algorithm
  correctly finds a tree of shortest paths.
\end{theorem}

\begin{proof}
  Let $s=w_0,w_1,\dots,w_k=v$ be a shortest path from the source $s$
  to a vertex $v$.  We prove by induction on $k$ that
  $d[w_k]=d_{w_k}$.  As $d[s]=d_s=0$, the claim is true for $k=0$.  Suppose
  $k>0$ and that our induction hypothesis is true for $k-1$.
  Since $w_k$ is reachable from the source, by Lemma~\ref{L-reachable}, it
  is adjoined to $S$ at some iteration, and we let $(u,w_k)$ denote
  the edge that is extracted from $P$ during that iteration.

  Suppose at first that $w_{k-1}\notin S$ when $(u,w_k)$ is extracted from $P$.  Then by Lemma~\ref{L-order}$\ref{P-extract}$,
  \[d[w_k]\leq d[w_{k-1}]=d_{w_{k-1}} = d_{w_k} - c[w_{k-1},w_k],\] which combined with
  Lemma~\ref{L-delta} implies $d[w_k]=d_{w_k}$ (and $c[w_{k-1},w_k]=0$).

  Suppose instead that $w_{k-1}\in S$ when $(u,w_k)$ is extracted from $P$.
  Since $(w_{k-1},w_k)$ is a shortest path edge, it is either
  out-pertinent or in-pertinent.  If it is in-pertinent, by Lemma~\ref{L-in-request} it was requested by the end of the previous iteration.
  In either case, by the end of the previous iteration,
  edge $(w_{k-1},w_k)$ is an out-pertinent or requested edge that has not been extracted from $P$,
  and $w_{k-1}\in S$.
  By Lemma~\ref{L-inP}, there must be some edge
  $(w_{k-1},x)$ in queue $P$ for which
\begin{align*}
  key[w_{k-1},x]=d[w_{k-1}]+c[w_{k-1},x]
 &\leq d[w_{k-1}]+c[w_{k-1},w_k]\\
 &=d_{w_{k-1}}+c[w_{k-1},w_k]=d_{w_k} \\
 &\leq d[u]+c[u,w_k]=key[u,w_k]=d[w_k].
\end{align*}
  Since it was edge $(u,w_k)$ that was extracted from $P$,
  $key[u,w_k]\leq key[w_{k-1},x]$, so the inequalities $\leq$ are
  actually equalities, and $d[w_k]=d_{w_k}$.

  Thus by induction each vertex $v$ reachable from the source satisfies $d[v]=d_v$,
  and so the algorithm computes a shortest path tree.
\end{proof}

\subsection{Complexity of the algorithm}\label{SS-complexity}

The running time of the algorithm is clearly dominated by the priority queue operations.
The following two lemmas show that the number of priority queue operations performed by the algorithm is $O(|\Eper|)$.

\begin{lemma}\label{L-request}
  For the forward-backward algorithm, all requested edges are in-pertinent.
\end{lemma}

\begin{proof}
Edges are requested from within the inner ``{\bf while} $Q\neq\varnothing$ \AND\ $\min(Q)< 2\,(\min(P)-M)$'' loop performed at the end of each iteration of the algorithm. Suppose an edge $(u,v)$ is requested.  At that time, the algorithm is in stage 2, so $M<\infty$.  If vertex $v$ is unreachable from the source~$s$, then $d_v=\infty$, so $c[u,v]<2\,(d_v-M)$ and edge $(u,v)$ is in-pertinent.  If on the other hand $v$ is reachable from the source, by Lemma~\ref{L-reachable} eventually $v$ is adjoined to~$S$, and as $v\notin S$ at the time that $(u,v)$ was requested, there is some vertex $v_1$ which is the next vertex adjoined to~$S$ upon the extraction of some edge $(u_1,v_1)$ from queue $P$.  At the time of $(u_1,v_1)$'s extraction, $u_1\in S$, so in fact $u_1\in S$ at the time edge $(u,v)$ is requested.  Furthermore, by the correctness of the algorithm (Theorem~\ref{T-correct}), $(u_1,v_1)$ is the last edge on a shortest path to~$v_1$.

Suppose at first that $(u_1,v_1)$ is either out-pertinent or was requested earlier than $(u,v)$. As $u_1\in S$ and as $(u_1,v_1)$ was not extracted yet from~$P$, Lemma~\ref{L-inP} says that~$P$ contains an edge $(u_1,v'_1)$ with $c[u_1,v'_1]\le c[u_1,v_1]$.  As $u_1\in S$, $\min(P)\le d_{v_1}$ at the time that $(u,v)$ is requested.  Since $v\notin S$ at the time of the request, we also have $d_{v_1}\leq d_v$.
Thus
\[c[u,v] \;<\; 2\,(\min(P)-M) \;\le\; 2\,(d_{v_1}-M) \;\le\; 2\,(d_{v}-M)\,,\]
so $(u,v)$ is in-pertinent as required.

Suppose now that $(u_1,v_1)$ is not out-pertinent and was not requested earlier than $(u,v)$.
As $(u_1,v_1)$ is an SPT edge that is not out-pertinent, it must be in-pertinent, so by Lemma~\ref{L-in-request}, it is requested by the time $(u_1,v_1)$ is extracted.
 Since edge $(u_1,v_1)$ is requested at the same time as or later than $(u,v)$, it was extracted from $Q$ at the same time as or later than $(u,v)$, so by Lemma~\ref{L-order}$\ref{Q-extract}$ $c[u,v]\le c[u_1,v_1]$.  Since $v\notin S$, by Lemma~\ref{L-order}$\ref{P-extract}$ and Theorem~\ref{T-correct} we have $d_v\geq d_{v_1}$.  So the in-pertinence of $(u_1,v_1)$ implies the in-pertinence of $(u,v)$.
\end{proof}

\begin{remark} \label{L-condition-true}
  Since in-pertinent edges are not out-pertinent, one consequence of
  Lemma~\ref{L-request} is that \request\ does not need to test
    (as in Figure~\ref{requestIf}) that the edge is not out-pertinent,
    so that the version of \request\ given in Figure~\ref{F-sssp} is correct.
\end{remark}

\begin{theorem}\label{T-PQ}
For the forward-backward algorithm,
\begin{itemize}
\item[$(i)$] The edges inserted into~$Q$ are all the in-pertinent edges, together with a lightest incoming non-in-pertinent edge (if any) to each vertex found after the $\lceil n/2\rceil$th vertex.
\item[$(ii)$] All edges inserted into~$P$, except possibly one outgoing edge for each vertex, are pertinent edges.
\end{itemize}
\end{theorem}

\begin{proof}
$(i)$.
By Lemmas~\ref{L-in-request} and~\ref{L-request}, the requested edges are precisely the in-pertinent edges.  For each such requested $(u,v)$, the next incoming edge to $v$ is inserted into $Q$.
The lightest incoming edge to those vertices found after the median are also inserted into $Q$.

  $(ii)$.  By Lemmas~\ref{L-delta} and~\ref{L-order}$\ref{P-extract}$ and Theorem~\ref{T-correct}, every edge $(u,v)$ extracted from~$P$ during the first stage of the algorithm satisfies $d_{u}+c[u,v] = key[u,v]\le M$.  Thus $c[u,v]\le M-d_{u}\leq 2\,(M-d_{u})$, so $(u,v)$ is out-pertinent.  Thus, only edges inserted, but not extracted, from~$P$ during the first stage may be non out-pertinent edges.  There can be at most one such outgoing edge for each vertex.
 During the second stage of the algorithm, edges inserted into~$P$ are either out-pertinent edges from~$Out[u]$, or requested edges, which by Lemma~\ref{L-request} are in-pertinent.
\end{proof}

\begin{theorem}\label{T-sssp} The running time of the forward-backward single-source shortest paths algorithm, when run on $\G_n(\EXP(1))$ with sorted adjacency lists, is $O(n)$, with very high probability.
\end{theorem}

\begin{proof}
 By Theorem~\ref{T-PQ}, the algorithm performs $O(|\Eper|+n)$ priority queue operations. By Theorem~\ref{T-pertinent-exp}, $|\Eper|$ is $\Theta(n)$, with very high probability, for  $\G_n(\EXP(1))$. We show in Section~\ref{S-bucket} that the priority queue operations performed by the algorithm can be implemented in constant amortized time per operation, again with very high probability.
\end{proof}

\section{Performance with exponential edge costs} \label{sec:probabilistic-analysis}

\subsection{Shortest path tree for randomly weighted graphs}\label{S-analysis}

For the complete graph with i.i.d.\ $\EXP(1)$ edge weights, Davis and Prieditis \cite{DaPr93} and Janson
\cite{Janson99} gave an elegant characterization of the set of all
distances from a given source vertex $s$, which we now recall.  Let
$v_1,\dots,v_n$ denote the vertices arranged in increasing order of
distance from the source $s$ (in particular $v_1=s$).  Let $d_u$ denote
the distance to vertex $u$ from the source.  For $k=2,\dots,n$, let
$p_k$ denote the index of $v_k$'s parent in the shortest-path tree,
i.e., $(v_{p_k},v_k)$ is an edge of the shortest path tree.
Because of the memoryless property of the exponential distribution,
and because there are $k(n-k)$ edges from $v_1,\dots,v_k$ to the
remaining $n-k$ vertices, it follows that $d_{v_{k+1}}-d_{v_k}$ is an
exponential random variable with mean $1/(k(n-k))$, independent of the
previous distances, $v_{k+1}$ is a uniformly random vertex from the
remaining vertices, and $p_{k+1}$ is a uniformly random choice from
$1,\dots,k$.  The quantities $v_{k+1}$, $p_{k+1}$, and
$d_{v_{k+1}}-d_{v_k}$ are mutually independent, and the only
dependence that they have upon the values of $v_1,\dots,v_k$,
$p_2,\dots,p_k$, and $d_{v_1},\dots,d_{v_k}$ is that $v_{k+1}$ is
distinct from $v_1,\dots,v_k$ and that $p_{k+1}$ is one of $v_1,\ldots,v_k$.

Let's write
\begin{align*}
X_k &\sim \EXP[1/(k(n-k))]\quad\quad\text{(independent of one another)}\\
d_{v_k} &= \sum_{i=1}^{k-1} X_i\,.
\end{align*}
The weight of any edge $(v_{p_k},v_k)$ of the shortest-path tree is
\[c(v_{p_k},v_k)=d_{v_k}-d_{v_{p_k}}\,.\]
For any edge $(v_j,v_k)$ not in the shortest path tree, we have
\[c(v_j,v_k) = \begin{cases}\EXP(1) & k<j \\ d_{v_k}-d_{v_j}+\EXP(1) & j<k\,,\end{cases}\]
where all these exponential random variables are mutually independent and also
independent of $v_1,\dots,v_n$, $d_{v_1},\dots,d_{v_n}$, and $p_2,\dots,p_n$.
(Here there is a slight difference between the cases of directed graphs
and undirected graphs.  In the case of undirected graphs, the formula
is $|d_{v_k}-d_{v_j}|+\EXP(1)$ regardless of the order of $j$ and $k$.
Otherwise, the characterizations of
$v_1,\dots,v_n$, $d_{v_1},\dots,d_{v_n}$, $p_2,\dots,p_n$ and $c(u,v)$
are the same for directed and undirected graphs.)

\subsection{Comparison of pertinent edges to a Poisson process}

The above characterization of the shortest path tree is particularly convenient for the purposes of comparing the distance to a vertex to the median distance,
allowing us to estimate the number of pertinent edges.

\begin{theorem} \label{thm:Poisson-domination}
Let $\Lambda$ denote the random variable
\[ \Lambda =
2\,(n-1) \sum_{k=\lceil n/2\rceil}^{n-1} \frac{1}{k} \EXP(1)\,,
\]
where the $\EXP(1)$'s are i.i.d.  Then for the directed
graph $\G_n(\EXP(1))$, the number of out-pertinent edges
that are not shortest path tree edges
is stochastically dominated by $\Poi(\Lambda)$, and similarly for the
number of in-pertinent edges that are not SPT edges.
\end{theorem}
\begin{proof}
For each edge we can associate an independent Poisson point process on
$\R^+$, and that edge's associated exponential random variable in its
weight is then just the first point in the point process.  The
indicator variable for the edge being lighter than a certain threshold
is dominated by the number of points of the point process in that
interval.  Thus the number of out-pertinent edges which are not
shortest-path tree edges is dominated by a Poisson random variable
with a random rate $\Lambda_\text{out}$, and similarly for the in-pertinent edges
which are not SPT edges.  The variable $\Lambda_\text{in}$ for in-pertinent
edges is given by
\begin{align*}
 \Lambda_\text{in} &= (n-1) \sum_{j=\lceil n/2\rceil+1}^{n} 2\,(d_{v_j}-M)\\
         &= 2\,(n-1) \sum_{j=\lceil n/2\rceil+1}^{n} \sum_{k=\lceil n/2\rceil}^{j-1} X_k \\
         &= 2\,(n-1) \sum_{k=\lceil n/2\rceil}^{n-1} (n-k) X_k \\
         &= 2\,(n-1) \sum_{k=\lceil n/2\rceil}^{n-1} \frac{1}{k} \EXP(1)\,,
\end{align*}
where the $\EXP(1)$'s are i.i.d., and similarly for out-pertinent edges the rate is
\begin{align*}
 \Lambda_\text{out} &= (n-1) \sum_{j=1}^{\lceil n/2\rceil-1} 2\,(M-d_{v_j})\\
         &= 2\,(n-1) \sum_{j=1}^{\lceil n/2\rceil-1} \sum_{k=j}^{\lceil n/2\rceil-1} X_k \\
         &= 2\,(n-1) \sum_{k=1}^{\lceil n/2\rceil-1} k X_k \\
         &= 2\,(n-1) \sum_{k=1}^{\lceil n/2\rceil-1} \frac{1}{n-k} \EXP(1)
         = 2\,(n-1) \sum_{k=\lfloor n/2\rfloor+1}^{n-1} \frac{1}{k} \EXP(1)\,.
\end{align*}
When $n$ is odd, $\Lambda_\text{out}$ and $\Lambda_\text{in}$ are equal in distribution, while
if $n$ is even, $\Lambda_\text{in}$ has one more term.  In either case, the theorem follows.
\end{proof}

\subsection{First moment estimate}

Using the comparison with a Poisson process we get an easy $O(n)$ bound on the expected number of pertinent edges.
We obtain the correct constant later in Section~\ref{expected-pertinent-constant}.

\begin{theorem} \label{T-EXP-Exp}
For $\Lambda$ defined in Theorem~\ref{thm:Poisson-domination}, we have
\[\E[\Poi(\Lambda)] < (\log 4)n+1\,.\]
\end{theorem}
\begin{proof}
It is straightforward that $\E[\Poi(\Lambda)]=\E[\Lambda]$.
We have
\[ \E[\Lambda] = 2\,(n-1) \sum_{k=\lceil n/2\rceil}^{n-1} \frac{1}{k}\,.
\]
For odd $n$ the sum is at most
\[\int_{(n-1)/2}^{n-1}\frac{dk}{k}=\log 2.\]
For even $n$ the sum is at most
\[\frac{1}{n-1}+\int_{(n-2)/2}^{n-2}\frac{dk}{k}=\frac{1}{n-1}+\log 2\,.\]
Thus for $n$ even, $\E[\Lambda]\leq(\log 4)(n-1)+2<(\log 4)n+1$.
\end{proof}
Thus the expected number of in-pertinent edges not in the shortest path tree is less than $(\log 4)n+1$, and similarly for out-pertinent edges not in the shortest path tree.
In particular, the expected number of pertinent edges is less than $(1+4\log 2)n + 1<3.7726n+1$.

\subsection{Large deviations}
\begin{theorem}\label{T-pertinent-exp}
  Let $|\Eper|$ be the number of pertinent edges with respect to a tree
  of shortest paths of (either directed or undirected) $\G_n(\EXP(1))$, where $M$ is taken to be the
  median distance.
  Then $\Pr[ |E_{per}|\ge (5+c)n ] < {\rm e}^{-n \frac{c^2}{10(5+c)}}$.
\end{theorem}
\begin{proof}
Let $\Lambda_{in} = 2\,(n-1) \sum_{k=\lceil n/2\rceil}^{n-1} \frac{1}{k} \EXP(1)$,
$\Lambda_{out} = 2\,(n-1) \sum_{k=\lfloor n/2\rfloor+1}^{n-1} \frac{1}{k} \EXP(1)$ and $\Lambda = \Lambda_{in} + \Lambda_{out}$. We have shown above that the number of pertinent edges that are not shortest paths edges is stochastically dominated by the random variable $\Poi(\Lambda)$. (Note that all $\EXP(1)$ random variables in the definition of $\Lambda$ are still independent.)
Define $\Lambda'=\sum_{i=1}^n 4\EXP(1)$. As $\frac{2(n-1)}{k}<4$, for $\lceil n/2\rceil\le k<n$, and as $\Lambda$ is the sum of at most $n-1$ terms, we get that~$\Lambda$ is stochastically dominated by $\Lambda'$.

It is a simple exercise to check that $\Poi(\alpha \EXP(1))$ is just $\Geo(\frac{1}{1+\alpha})$, where $\Geo(p)$ is a \emph{geometric} random variable which counts the number of \emph{failures} until the first success in a sequence of independent Bernoulli trials, each with a success probability of~$p$. (Note that according to this definition, the success ending the experiment is not counted. Thus, $\Geo(p)$ may attain the value~$0$, and $\E[\Geo(p)]=\frac{1}{p}-1$.)

Thus, $\Poi(\Lambda)$ is stochastically dominated by $\Poi(\Lambda')$ which has the same distribution as the sum of $n$ independent copies of $\Geo(\frac{1}{5})$. This, in turn, is just $\NB(n,\frac{1}{5})$, where $\NB(r,p)$ is a \emph{Negative Binomial} random variable which counts the number of \emph{failures} until achieving $r$ successes in a sequence of independent Bernoulli trials, each with a success probability of $p$.

Now, $\Pr[ \NB(r,p) \ge k ] \le \Pr[ \BIN(r+k,p)\le r]$, where $\BIN(n,p)$ is a \emph{Binomial} random variable. If $X\sim \BIN(n,p)$, so that $\mu=\E[X]=np$, then by Chernoff's bound we have $\Pr[X<(1-\delta)\mu] \le {\rm e}^{-\mu\delta^2/2}$.

Thus, for $k=(4+c)n$, where $c>0$, we have
\[ \textstyle \Pr[ |E_{per}|\ge (5+c)n ] \;\le\; \Pr[ \Poi(\Lambda') \ge k ] \;=\;
\Pr[ \NB(n,\frac{1}{5}) \ge k ] \;\le\; \Pr[ \BIN(n+k,\frac{1}{5}) \le n ]\;.
\]
Let $X\sim \BIN(n+k,\frac{1}{5})$. We have $\mu = \frac{1}{5}(n+k) = \frac{5+c}{5}n$, and $n=(1-\delta)\mu$, where $\delta=\frac{c}{5+c}$. Thus,
\[ \textstyle \Pr[ |E_{per}|\ge (5+c)n ] \;<\; {\rm e}^{-\mu\delta^2/2}
\;=\; {\rm e}^{-\frac{5+c}{5}n\, \left(\frac{c}{5+c}\right)^2/2} \;=\;
{\rm e}^{-n \frac{c^2}{10(5+c)}}
\;. \qedhere\]
\end{proof}

\subsection{Expected number of pertinent edges} \label{expected-pertinent-constant}

Since \[\E[X_k] = \frac{1}{k(n-k)} = \frac{1}{n}\left(\frac{1}{k} + \frac{1}{n-k}\right) \,,\]
by induction we have
\[
 \E\big[d_{v_k}\big] = \frac{H_{k-1} - H_{n-k} + H_{n-1}}{n}\,,
\]
where $H_k=1/1+\cdots+1/k = \log k + \gamma + 1/(2k) + \cdots$ is the $k$-th harmonic number
and $\gamma=0.577\dots$ is Euler's constant.
This is approximately
\[
 \E\big[d_{v_k}\big] = M + \frac{1}{n}\log\frac{k}{n+1-k} + O(1/(kn)) + O(1/((n+1-k)n))\,,
\]
where $M=d_{v_{\lceil n/2\rceil}}$ is the median distance.

For integers $i$ and $j$ in the range from $1$ to $n$ we let \[Y_{i,j}=2M-d_{v_i}-d_{v_j}\,.\]

\begin{lemma}
For $\G_n(\EXP(1))$, in either the directed or undirected setting, for $1\leq i<j\leq n$,
\begin{align*}
\Pr[\text{edge $(v_i,v_j)$ is non-SPT out-pertinent}] &= \tfrac{j-2}{j-1}\,\E[\max(0,1-\exp(-Y_{i,j}))] \\
\Pr[\text{edge $(v_i,v_j)$ is non-SPT in-pertinent}] &= \tfrac{j-2}{j-1}\,\E[\max(0,1-\exp(+Y_{i,j}))]\,.
\end{align*}
\end{lemma}
\begin{proof}
If $i\geq \lceil n/2\rceil$ then a.s.\ edge $(v_i,v_j)$ is not
out-pertinent, and if $j>i\geq \lceil n/2\rceil$ then
$Y_{i,j}\leq0$ so $\E[\max(0,1-\exp(-Y_{i,j}))]=0$. Suppose now that
$i<\lceil n/2\rceil$ and $i<j$.  With probability $1/(j-1)$ edge
$(v_i,v_j)$ is an SPT edge, and otherwise $c(v_i,v_j)\sim
d_{v_j}-d_{v_i}+\EXP(1)$. Conditional on edge $(v_i,v_j)$ not being
an SPT edge, it is out-pertinent when
\begin{align*}
d_{v_j}-d_{v_i}+\EXP(1) &\leq 2\,(M-d_{v_i}) \\
\EXP(1) &\leq 2M-d_{v_i}-d_{v_j} = Y_{i,j}\\
\Pr[\text{$(v_i,v_j)$ out-pertinent} \,|\, \text{distances, not SPT edge}] &= \max(0,1-\exp(-Y_{i,j}))\\
\Pr[\text{$(v_i,v_j)$ out-pertinent}\,|\, \text{not SPT edge}] &= \E[\max(0,1-\exp(-Y_{i,j}))]\,.
\end{align*}

The in-pertinent edges are similar.  We need only consider $j>\lceil n/2\rceil$, and we are checking if $c(v_i,v_j)<2\,(d_v-M)$.
For $i<j$ we are testing if
\begin{align*}
d_v-d_u+\EXP(1) &< 2\,(d_v-M) \\
\EXP(1) &< d_v+d_u - 2M = -Y_{i,j}\\
\Pr[\text{$(v_i,v_j)$ in-pertinent}\,|\, \text{not SPT edge}] &= \E[\max(0,1-\exp(+Y_{i,j}))]\,. \qedhere
\end{align*}
\end{proof}

\begin{lemma}
For any real-valued random variable $X$
\[
\max(0,\E[X]) \leq \E[\max(0,X)] \leq \max(0,\E[X])
 + \Var[X]^{1/2}\,.
\]
\end{lemma}
\begin{proof}
Since $\max(0,\cdot)$ is convex,
\[\E[\max(0,X)]\geq\max(0,\E[X])\,.\]
In the reverse direction,
\[\E[\max(0,X)] \leq \E[|X|] \leq \E[X^2]^{1/2} = (\E[X]^2+\Var[X])^{1/2} \,.\]
Observe that $(a+b)^{1/2} \leq a^{1/2} + b^{1/2}$.
For $\E[X]\geq 0$ we have
\[ \E[\max(0,X)] \leq \E[X] + \Var[X]^{1/2} \,,\]
and for $\E[X]\leq 0$ we have $\E[\max(0,X)] = \E[X] + \E[\max(0,-X)] \leq \Var[X]^{1/2}$.  Combining these inequalities gives the lemma.
\end{proof}

\begin{theorem} \label{T-expected-pertinent}
For the undirected graph $\G_n(\EXP(1))$,
\begin{align*}
\E[\text{\# out-pertinent non-SPT edges}] &= (\tfrac12 + o(1)) n\\
\E[\text{\# in-pertinent non-SPT edges}] &= (\tfrac12 + o(1)) n\,,
\end{align*}
while for the directed graph $\G_n(\EXP(1))$,
\begin{align*}
\E[\text{\# out-pertinent non-SPT edges}] &= (\log 2 + o(1)) n\\
\E[\text{\# in-pertinent non-SPT edges}] &= (\log 2 + o(1)) n\,.
\end{align*}
\end{theorem}
\begin{proof}
For real $x$,
\[
\max(0,x)-x^2/2 \leq \max(0,x-x^2/2) \leq \max(0,1-\exp(-x)) \leq \max(0,x)\,.
\]
Thus
\[
\max(0,\E[X]) - \tfrac12 \E[X^2] \leq \E[\max(0,1-\exp(-X))] \leq \max(0,\E[X]) + \Var[X]^{1/2}\,.
\]
We apply these bounds to the random variables $Y_{i,j}$ and $-Y_{i,j}$.

For $i\leq\lceil n/2\rceil\leq j$ we have
\begin{align*}
Y_{i,j}=2M-d_{v_i}-d_{v_j} &= X_i+\cdots +X_{\lceil n/2\rceil-1} - X_{\lceil n/2\rceil} - \cdots - X_{j-1}\\
\Var[Y_{i,j}] &= \sum_{k=i}^{j-1} \frac{1}{k^2 (n-k)^2} \\
\Var[Y_{i,j}] &\leq O\left(\frac{1/i+1/(n+1-j)}{n^2}\right)
\end{align*}
For $i,j\leq\lceil n/2\rceil$ or $\lceil n/2\rceil \leq i,j$ the exact formula
for $\Var[Y_{i,j}]$ differs, but the bound in the last equation still holds.

\begin{align*}
\Var[Y_{i,j}]^{1/2} &\leq O\left(\frac{1/i^{1/2}+1/(n+1-j)^{1/2}}{n}\right)\\
\sum_{i,j} \Var[Y_{i,j}]^{1/2} &\leq O(\sqrt{n})\,.
\end{align*}

We also have $\E[Y_{i,j}]\leq O((\log n) / n)$, so
\begin{align*}
\E[Y_{i,j}^2] &= \Var[Y_{i,j}] + \E[Y_{i,j}]^2
\leq O((\log^2 n)/n^2)\,.
\end{align*}

So we deduce that the expected number of non-SPT out-pertinent
edges $(v_i,v_j)$ with $i<j$ is
\[\sum_{i=1}^{\lceil n/2\rceil-1} \sum_{j=i+1}^{n} \tfrac{j-2}{j-1}\,\max(0,\E[Y_{i,j}]) + O(\sqrt{n})\,.
\]
and that the expected number of non-SPT in-pertinent edges $(v_i,v_j)$ with $i<j$ is
\[\sum_{j=\lceil n/2\rceil+1}^{n} \sum_{i=1}^{j-1} \tfrac{j-2}{j-1}\,\max(0,-\E[Y_{i,j}]) + O(\sqrt{n})\,.
\]

But
\[
\E[Y_{i,j}] = \frac{1}{n} \log\frac{n+1-i}{i} + \frac{1}{n} \log\frac{n+1-j}{j} + O(1/(n\min(i,j,n+1-i,n+1-j)))\,,
\]
and the error term adds up to at most $O(\log n)$, so the non-SPT out-pertinent sum is
\[ \frac{1}{n} \sum_{i=1}^{\lceil n/2\rceil-1} \sum_{j=i+1}^{n-i} \frac{j-2}{j-1}\, \left(\log\frac{n+1-i}{i} + \log\frac{n+1-j}{j}\right) + O(\sqrt{n})\,.
\]
Since
\[ \frac{1}{n} \sum_{1\leq i < j \leq n} \frac{1}{j-1} \; 2\log n \leq O(\log^2 n)\,,
\]
we can drop the factor of $(j-2)/(j-1)$, and since
\[ \sum_{j=i+1}^{n-i} \log\frac{n+1-j}{j} = 0\,,
\]
the above sum simplifies to
\[ \frac{1}{n} \sum_{i=1}^{\lceil n/2\rceil-1} (n-2i) \log\frac{n+1-i}{i} + O(\sqrt{n})\,.
\]
The summand is monotone, so we may replace the sum with
\[
 \int_{0}^{1/2} (1-2x) \log\frac{1-x}{x}\times n\, dx + O(\log n)= \frac{n}{2} + O(\log n)\,.
\]
Up to negligible error terms, the calculations for the non-SPT in-pertinent edges are symmetric
to those for the non-SPT out-pertinent edges, and also give $n/2+O(\sqrt{n})$.

In the directed setting, we also need to consider edges $(v_i,v_j)$ for which $i>j$.
If $j<i$ then $c(v_i,v_j)\sim\EXP(1)$ which is independent of $M-d_{v_i}$,
so for $j<i$
\begin{align*}
\Pr[\text{edge $(v_i,v_j)$ is non-SPT out-pertinent}] &= \E[\max(0,1-\exp(-Y_{i,i}))] \\
\Pr[\text{edge $(v_i,v_j)$ is non-SPT in-pertinent}] &= \E[\max(0,1-\exp(+Y_{j,j}))]\,.
\end{align*}
Using the same estimates as above, we find that the expected number of such out-pertinent
edges is
\[
\frac{1}{n} \sum_{i=1}^{\lceil n/2\rceil-1} \sum_{j=1}^{i-1} 2\log\frac{n+1-i}{i} + O(\sqrt{n}) = (\log 2 -\tfrac12) n + O(\sqrt{n})
\]
and similarly for the expected number of such in-pertinent edges.
\end{proof}

\begin{theorem}
For either the directed or undirected graph $\G_n(\EXP(1))$,
\begin{align*}
\E[\text{\# out-pertinent SPT edges}] &= (\log 2 + o(1)) n\\
\E[\text{\# in-pertinent SPT edges}] &= (1-\log 2 + o(1)) n\,.
\end{align*}
\end{theorem}
\begin{proof}
There are $n-1$ SPT edges, each of which is either out-pertinent or in-pertinent but not both.
For $i<j$, edge $(v_i,v_j)$ is an SPT edge with probability $1/(j-1)$, and if so,
it is out-pertinent when $Y_{i,j}\geq 0$ and in-pertinent when $Y_{i,j}<0$.
When $i,j\leq\lceil n/2\rceil$ or $i,j\geq\lceil n/2\rceil$, the sign of $Y_{i,j}$ is almost
surely positive, or negative, respectively, so we are left with the case $i<\lceil n/2\rceil < j$.
We have already computed $\E[Y_{i,j}]$ and $\Var[Y_{i,j}]$ in this case,
so we can use Chebychev's inequality to estimate $\Pr[Y_{i,j}\geq 0]$.
For ``most'' pairs $i<j$ we have $\Var[Y_{i,j}]=\Theta(1/n^3)$ while $|\E[Y_{i,j}]|=\Theta(1/n)$,
so in expectation a negligible fraction of the pairs $i<j$ are such that the sign of $Y_{i,j}$
differs from the sign of $n+1-i-j$.  The expected number of in-pertinent edges is thus
\[
\sum_{j=\lceil n/2\rceil+1}^n\; \sum_{i=n+2-j}^{j-1} \frac{1}{j-1} + o(n) = \int_{1/2}^1 \frac{2x-1}{x}\times n\,dx + o(n)\,.\qedhere
\]
\end{proof}

\section{Performance with uniform and other edge costs}\label{sec:other}

In this section we show how results similar to the ones obtained in Sections~\ref{sec:probabilistic-analysis}
for the \emph{exponential\/} distribution can be obtained for similar distributions, e.g., the \emph{uniform\/} distribution.

Let $X\sim \EXP(1)$ and let $Y=f(X)$, where $f:[0,\infty)\to[0,\infty)$ is monotone increasing. For example, if $f_U(x)=1-{\rm e}^{-x}$, then $f_U(X)\sim U[0,1]$, i.e., a random variable uniformly distributed in $[0,1]$. Note that $f_U(x)=x+O(x^2)$, as $x\to 0$.

Let $\G_n(Y)$ be the probabilistic model in which independent random edge weights, each identical in distribution to~$Y$, are assigned to the edges of a complete graph on $n$ vertices.

\begin{theorem}
  If $Y=f(\EXP(1))$, where $f(x)=x+O(x^2)$ as $x\to0$, then for
  $\G_n(Y)$ we have $\E[|\Eper|]=O(n)$, and for some $c>1$ and
  $\alpha>0$ we have $\Pr[|\Eper|>cn+\Delta]< \ee^{-\alpha\Delta} +
  O(n^3 \ee^{-n^{1/2}})$.
\end{theorem}

Note that the main loss compared to Theorem~\ref{T-pertinent-exp} is the additive term $O(n^3\ee^{-n^{1/2}})$ in the probability bound.
We conjecture that this additive term can be removed.

\begin{proof}
Let $c[u,v]$ and $d_v$ be the edge weights and distances in $\G_n(\EXP(1))$.  Let $c'[u,v]=f(c[u,v])$ be the weight of edge $(u,v)$ in
$\G_n(Y)$, and let $d'_v$ be the distance from the source to $v$ in $\G_n(Y)$.

For $\G_n(\EXP(1))$, Janson \cite{Janson99} proved that for any $a>2$,
\[\Pr\bigg[\max_v d_v> a\log n/n\bigg] \leq O(n^{2-a} \log^2 n)\,.\]
We use this result with $a=n^{1/2}/\log n$. It follows that the probability that $d_v>n^{-1/2}$ for some vertex~$v$ is at most $O(n^{2-a}\log^2 n) < O(n^3 \ee^{-n^{1/2}})$. We thus assume in the sequel that $d_v\le n^{-1/2}$ for every vertex~$v$.

Suppose now that $|f(x)-x|\le \beta\min\{x^2,f(x)^2\}$, for some fixed $\beta>0$, for every $x$ sufficiently close to~$0$. (The existence of~$\beta$ follows from our assumption that $f(x)=x+O(x^2)$.)

Let $s=v_0,v_1,\ldots,v_\ell=v$ be a shortest path from~$s$ to~$v$ in $\G_n(\EXP(1))$.
Let $c_1,c_2,\ldots,c_\ell$ be the weights of its edges in $\G_n(\EXP(1))$, and let $c'_1,c'_2,\ldots,c'_\ell$ be the corresponding edge weights in $\G_n(Y)$. Using the fact that $d_v=\sum_{i=1}^\ell c_i$ and the convexity of the function $x^2$ we get
\[ d'_v \;\le\; \sum_{i=1}^\ell c'_i \;\le\;
   \sum_{i=1}^\ell c_i + \beta \sum_{i=1}^\ell c_i^2 \;\le\; d_v + \beta d_v^2 \;\le\; d_v+\beta n^{-1}\,.
\]
Similarly, by looking at a shortest path from~$s$ to~$v$ in $\G_n(Y)$, we get, for sufficiently large~$n$, that
\[ d_v \;\le\; d'_v + \beta {d'_v}^2 \;\le\; d'_v + \beta(d_v+\beta n^{-1})^2 \;\le\; d'_v+(1+o(1))\beta n^{-1} \,, \]
where the $o(1)$ term goes to $0$ as $n\to\infty$.
Combining these two inequalities, we get that
\[ |d'_v-d_v| \le (1+o(1))\beta n^{-1}\,. \]
Let $M$ and $M'$ be the median distances in $\G_n(\EXP(1))$ and $\G_n(Y)$, respectively.
As $d_v\le M$ for $n/2$ vertices, we have $d'_v\le M+\beta n^{-1}$ for these $n/2$ vertices,
and thus $M'\le M + \beta n^{-1}$. In a similar way we get that $M\le M' + (1+o(1))\beta n^{-1}$ and thus
\[|M'-M|\le (1+o(1))\beta n^{-1}\,.\]

Now suppose $(u,v)$ is an out-pertinent edge in $\G_n(Y)$, i.e., $c'[u,v]\le 2\,(M'-d'[u])$.  Then $c'[u,v]\leq 2n^{-1/2} + O(n^{-1})$, so $c[u,v]\leq c'[u,v] + (4+o(1))\beta n^{-1}$, and hence
 $c[u,v] \le 2\,(M-d[u]) + (8+o(1))\beta n^{-1}$. In other words, $(u,v)$ is \emph{almost\/} out-pertinent in $\G_n(\EXP(1))$. More formally, we say that and edge $(u,v)$ is $\delta$-out-pertinent, if and only if $c[u,v]\le 2\,(M-d[u])+\delta$. If follows that if $(u,v)$ is out-pertinent in $\G_n(Y)$, then it is $(8+o(1))\beta n^{-1}$-out-pertinent in $\G_n(\EXP(1))$. The same observations and definitions apply to in-pertinent edges.

It is not difficult to see that for any $\delta>0$, the number of $\delta$-out-pertinent edges that are not out-pertinent is stochastically dominated by a Poisson process with rate $\delta n^2$.  The number of non-SPT edges which are $\delta$-in-pertinent but not in-pertinent is also dominated by a Poisson process with rate $\delta n^2$.  We have $\delta n^2=(8+o(1))\beta n$, so the probability that there are $16\beta n$ such edges is exponentially small in~$n$. This small probability is absorbed in the $O(n^3\ee^{-n^{1/2}})$ term.
\end{proof}

\section{Performance with Weibull edge costs}\label{sec:Weibull}

In our analysis of the performance of the algorithm for randomly
weighted graphs we considered edge costs that are given by i.i.d.\
exponential random variables.  We could consider other distributions
for the edge costs, but still have i.i.d.\ edge weights.  Since it is
only the very light edges that are involved in the shortest path tree,
it is conventional wisdom that only the probability density function
at $0$ is really important, and as long as this is positive, the
shortest path tree and distances will be just a perturbation of the
case when the edge costs are exponential random variables.

Bhamidi and van der Hofstad \cite{MR2932542} studied shortest paths on
the complete graph with i.i.d.\ edge weights with a more general class
of edge weight random variables, where the probability density
function at $0$ is either infinite or zero.  In these cases, the
shortest path tree is genuinely different from, i.e., not merely a
perturbation of, the tree with exponential edge costs.  More
precisely, Bhamidi and van der Hofstad considered edge costs given by
$\EXP(1)^s$, where $s>0$ is a parameter.  $\EXP(1)^s$ is a \textit{Weibull\/}
random variable with shape parameter $1/s$.  Near $0$, the
probability density function scales like $\frac{1}{s} x^{1/s-1}$.
Bhamidi and van der Hofstad studied the cost and hopcount of the
shortest path connecting typical vertices \cite{MR2932542}.  Here we
consider this same edge cost model, and study the performance of the
algorithm on it.

\begin{lemma}
  For fixed $s>0$, with i.i.d.\ Weibull $\EXP(1)^s$ edge costs on
  $\G_n$, the expected number of out-pertinent edges is $O(n)$.
\end{lemma}
\begin{proof}
  We use the analysis of Bhamidi and van der Hofstad \cite{MR2932542}.
  They showed how to couple an out-exploration process (scaled by
  $n^s$) from one vertex and an in-exploration process (also scaled by
  $n^s$) from another vertex to two independent branching processes,
  so that the branching processes stochastically dominate the
  exploration processes, and are usually the same, until $O(\sqrt{n})$
  vertices are in the exploration processes.  At some point the two
  exploration processes merge, at which point the shortest path
  connecting the vertices is determined.  Almost surely the branching
  processes have asymptotically $W_s\times \lambda_s^x$ vertices
  within distance $x$ of the root, as $x\to\infty$, where
  $\lambda_s=\Gamma(1+1/s)^s$, and $W_s$ is a certain random variable
  \cite{MR2932542}.  Bhamidi and van der Hofstad showed that the
  typical distance between pairs of vertices is (in the limit of large
  $n$) \[\frac{\log n + Y}{\lambda_s n^s}\,,\] where $Y$ is a certain
  random variable (independent of $n$).

  The same analysis gives a similar expression for the median distance
  $M$ from a source vertex $s$, but with a different random variable.
  Furthermore, the expected number of vertices~$u$ for which
  $n^s(M-d[u])>t$ is $O(n/\lambda_s^t)$.  The expected number of
  outgoing edges of weight $\leq 1/x^s$ from a vertex is
  $(n-1)(1-\ee^{-x})$.  The number of light outgoing edges from a vertex
  is an increasing function of the edge weights,
  and the value of the median distance is a decreasing function of the edge weights,
  so we can use the FKG inequality \cite{FKG}
   to bound the expected number of
  out-pertinent edges from a vertex by $O(1)$.
\end{proof}

\begin{lemma}
  For fixed $s$ such that $0<s\leq 1$, with i.i.d.\ Weibull
  $\EXP(1)^s$ edge costs on $\G_n$, the expected number of
  in-pertinent edges is $O(n)$.
\end{lemma}
\begin{proof}
  A Weibull random variable $\EXP(1)^s$ has the distribution of the
  first point of a Poisson point process with
  rate \[x\mapsto\frac{1}{s} x^{1/s-1}\,.\]
  For $s\leq 1$, this rate is monotonically increasing in $x$.
  Consider the shortest path tree up to the median distance $M$.
  For any vertex $v$ for which $d[v]>M$, we have that $d[v]-M$ is
  distributed according to the first point of a Poisson point process
  with rate at least \[\frac{n}{2s} x^{1/s-1}\,,\]
  so the expected number of vertices for which $d[v]-M\geq x$ is
  \[ \leq \frac{n}{2} \exp(-(n/2) x^{1/s})\,.\] For a given vertex $w$,
  the expected number of incoming edges with weight less than $2x$ is
  $(n-1) (1-\exp(-(2x)^{1/s}))$, which is approximately and at most $n
  (2x)^{1/s}$.  As above, we can use the FKG inequality,
  since $d[v]-M$ is an increasing function of the incoming edge weights,
  while the number of
  short incoming edges is a decreasing function of these edge weights.
  Integrating, we find that the expected
  number of in-pertinent edges is $O(n)$.
\end{proof}

We believe that the expected number of in-pertinent edges is $O(n)$ even if $1<s<\infty$.
The main technical issue is to bound $\Pr[n^s(d[v]-M)>t]$ for large $t$.

\section{All-pairs shortest paths}\label{S-all-pair}

The all-pairs shortest paths problem can be solved by running the
single-source shortest paths algorithm from each vertex.  The expected
APSP run time will be $n$ times the expected SSSP run time, and the
probability of a large deviation in the run time is at most multiplied
by a factor of $n$, which is negligibly affects the exponentially
small probability of a large deviation.  In this section we discuss
the sorting of the adjacency lists of the graph by edge cost, which
for the SSSP problem we assumed had been part of the preprocessing
phase.  We show that when the edge cost distribution is known, these
sorted adjacency lists can be constructed in $\Theta(n^2)$ time, with
exponentially small probability of a large deviation in the run time.

Since the edge cost distribution is known, we can use bucket sorting
to sort the edges by cost, and then scan through the sorted edges to
build up the outgoing and incoming adjacency lists.  Let $N=\binom{n}{2}$
denote the number of edges.  We can take the number of buckets to be $N$,
and set the range of each bucket so as to make each bucket equally likely
to contain a random edge cost.  Each bucket will then contain an approximately Poisson
number of edges.  These edge costs can then be sorted using, for example, heapsort.

\begin{theorem}
  Suppose the $N=\binom{n}{2}$ edge costs are sorted using bucket sort with $N$ buckets on which heapsort is used,
  as described above.  Then the run time has expected value $O(n^2)$, and there are constants $c_1$ and $c_2$ such that
  for any $\Delta\geq 0$,
  $\Pr[\text{run time}\geq c_1 n^2 + \Delta] \leq \exp(-c_2 \Delta)$.
\end{theorem}
\begin{proof}
  Suppose that the $i$th bucket contains $k_i$ items.  The time needed to heapsort these items is proportional to $1+\log k_i!$.
  The $1$'s of course add up to $N$, so we let $T_i=\log k_i!$, let $T=\sum_i T_i$, and bound the probability that $T$ is large.
  Of course
  \[ \Pr[T\geq t] = \Pr[\ee^T\geq \ee^t] \leq \ee^{-t} \, \E[\ee^T]\,.\]
  We can evaluate $\E[\ee^T]$ as a multinomial sum:
\begin{align*}
\E[\ee^T]
  &= \sum_{k_1,\dots,k_N} \binom{N}{k_1,\dots,k_N} \frac{1}{N^N} \ee^{\sum_i \log k_i!} \\
  &= \sum_{k_1,\dots,k_N} \frac{N!}{N^N} \\
  &= \binom{2N-1}{N} \frac{N!}{N^N} \\
  &= (1+o(1)) \frac{1}{\sqrt{2}} (4/\ee)^N\,,
\end{align*}
from which the large deviation claim follows.  The expected value claim is a consequence of the large deviation claim.
\end{proof}

\section{Efficient priority queues}\label{S-bucket}

Traditional comparison-based priority queues, such as the binary heaps of Williams \cite{Williams64}, require $O(\log n)$ time per heap operation, which is best possible given the $\Omega(n\log n)$ lower bound for comparison-based sorting.  The fastest implementations of Dijkstra's algorithm use priority queues like Fibonacci heaps \cite{FrTa87} that support \extractmin\ operations in $O(\log n)$ amortized time, and all other heap operations, including \decreasekey, in $O(1)$ amortized time. Our algorithm uses only \INSERT, \minkey\ and \extractmin\ operations --- in particular, it does not use \decreasekey\ operations.

Faster priority queues may be obtained by abandoning the comparison-based approach and assuming that keys are integers contained in single machine words. The fastest priority queues in this word RAM model require $O(\sqrt{\log\log n})$ expected time per operation (see Thorup \cite{Thorup07}).
Dial \cite{Dial69}, Ahuja \textit{et al.} \cite{AhMeOrTa90}, and Cherkassky \textit{et al.} \cite{ChGoSi99}
suggested bucket-based monotone priority queues.  We describe these in Section~\ref{sub:bucket};
they are somewhat slower asymptotically, but more practical.
We describe in Section~\ref{sub:two-level} a simple adaptation of the bucket-based priority
queues, and then in Section~\ref{sub:queue-deviation} we show that, in our probabilistic setting,
these priority queues require only $O(1)$ amortized time per operation, with high
probability, allowing us to implement algorithm \sssp\ of
Section~\ref{S-SSSP} in $O(n)$ time.

\subsection{Bucket-based monotone priority queues}\label{sub:bucket}

We briefly sketch a standard bucket-based implementation of a monotone
priority queue.  If the priority queue is comprised of $B$ buckets of
width $W$, then the $i$-th bucket, where $0\le i<B$, contains items
whose keys are in the interval $[i W,(i+1)W)$.  If there is no
\textit{a~priori\/} upper limit on the value of an item's key, then
any items with key $\geq B W$ can be placed in the last bucket.  The
items in each bucket could be stored in a linked list, a binary heap,
or another data structure which implements a monotone priority queue.
Initially all buckets are empty. One of the buckets is designated as
the \emph{active\/} bucket.  We let $a$ be the index of the active
bucket. Initially $a=0$.

The \INSERT\ and \extractmin\ operations are implemented as follows.
To insert an item~$x$ with key~$k$, we simply insert~$x$ into bucket
$\lfloor k/W \rfloor$.  To extract an item of minimum key, we check
whether the current active bucket is non-empty.  If it is, we perform
an extract-min operation on the bucket's priority queue, and return an
item of minimum key.  If the active bucket is empty, we increment~$a$
and repeat. The correctness of the implementation follows from the
monotonicity assumption.

It is not difficult to check that the total time required to perform a
sequence of operations on such a bucket-based monotone priority queue
is $O(B+\sum_i T_i)$, where~$T_i$ is the amount of work required to do
the \INSERT\ and \extractmin\ operations in bucket $i$'s sub-priority
queue.  The parameter $B$ is chosen to be large enough to make the
sub-queues small and hence the $T_i$'s small, but not so large as to
dominate the total running time.  In many applications a simple linked
list is sufficient for the sub-queues because they typically contain
so few elements.  For our purposes we need to use a more efficient
data structure for the priority queues associated with each bucket.

\subsection{Splitting active buckets into binary heaps}\label{sub:two-level}

In a bucket-based monotone priority queue, once a bucket becomes
active, it can itself be implemented by a bucket-based monotone
priority queue, and such multilevel buckets were suggested by
Cherkassky \textit{et al.} \cite{ChGoSi99}.  Elements of each
non-active top-level bucket can be stored in a simple unsorted linked
list until the bucket becomes active.

For the purposes of the forward-backward SSSP algorithm, we advocate
using two-level buckets, where the number of sub-buckets per top-level
bucket is not fixed ahead of time, but is instead determined by the
number of items in the top-level bucket at the time it becomes active.
When a top-level bucket becomes active, if it contains $b$ items, then
we split it into $b$ sub-buckets, each of which is in turn implemented
by a binary heap.
We can afford to sequentially scan these low level buckets to find the
first non-empty low-level bucket, since there are $b$ items over which
to amortize the scanning cost.  If the items in the top-level bucket
are approximately uniformly distributed, then the sub-buckets should
typically have $O(1)$ items.  In the course of running the
forward-backward algorithm, additional edges may be subsequently
inserted into the active top-level bucket.  However, there is no need
to resplit or rebalance the sub-buckets.  We rely on the fact that the
underlying binary heaps are efficient, and that few items are ever
inserted into a top-level bucket after it becomes active.
If there are any items whose key would fall outside the range of the top-level
buckets, these items are placed into the last top-level bucket.  In the event that
the last top-level bucket becomes active, its items are inserted
into a single binary heap (without any sub-bucketing).

Since binary heaps are used for the sub-buckets, the worst-case
performance of these priority queues is $\Theta(\log n)$ time per
examined edge.  As we shall see in the next section, when we take the top-level bucket width
$W$ to be $\Theta(1/(n\log n))$ and the number of buckets~$B$ to be $\Theta(n)$, when the forward-backward
SSSP algorithm is run on $\G_n(\EXP(1))$, not only is the expected
total running time $\Theta(n)$, but the probability that the total
running time exceeds $\Theta(n)$ is $\exp(-\Theta(n/\log n))$.

\subsection{Performance of the buckets of heaps}\label{sub:queue-deviation}

The priority queue $P$ is used to process the pertinent edges, and up to $n$ additional non-pertinent edges may also be inserted into $P$.  There are three main types of pertinent edges: shortest path tree edges, out-pertinent edges which are not SPT edges, and in-pertinent edges which are not SPT edges.  The keys for the SPT edges are just the distances of the vertices from the source, which have a simple characterization.  The keys of the out-pertinent edges which are not SPT edges is dominated by a Poisson process with random intensities depending on the SPT, and similarly for the in-pertinent edges which are not SPT edges.  The non-pertinent edges inserted into $P$ can also be characterized in terms of the Poisson process.

When a top-level bucket becomes active, if it contains $b$ items then it is evenly divided into $b$ sub-buckets which are each implemented by a binary heap.  Assuming the top-level buckets are small enough, the edges in them will be approximately uniformly distributed, so each sub-bucket will contain a small random number of items, and it will be unlikely for the binary heap operations to take a long time.  Also assuming the top-level buckets are small enough, there will be very few edges inserted into the active top-level bucket, and these will contribute only a small amount to the running time.

We already have exponentially tight bounds on the number of edges inserted into $P$.  Assuming that not too many edges are inserted into $P$, then their distribution into the top-level buckets is almost irrelevant to the performance of the priority queue.

Let $b_i$ denote the number of items in top-level bucket $i$ at the time that it becomes active.  Top-level bucket $i$ is then split into $b_i$ sub-buckets; let $J_{i,j}$ denote the number of items in the $j$th sub-bucket, where $0\leq j<b_i$.  If no further items are inserted into this sub-bucket, then the amount of time required to perform the $J_{i,j}$ \INSERT\ and \extractmin\ binary heap operations is proportional to $J_{i,j}+T_{i,j}$ where $T_{i,j}=\log(J_{i,j}!)=(1+o(1)) J_{i,j}\log J_{i,j}$.  Let $T=\sum_i\sum_j T_{i,j}$.  We wish to show that $\E[T]=O(n)$, and that the distribution of $T$ has exponentially decaying tails, and that the additional edges inserted into a top-level bucket after it becomes active do not contribute much to the runtime.
To obtain the exponentially decaying tails, we use the usual strategy for Chernoff-type bounds and write for $\alpha>0$
\begin{align*}
\Pr[T\geq t] &= \Pr[\ee^{\alpha T}\geq \ee^{\alpha t}] \\
&\leq \ee^{-\alpha t} \,\E[\ee^{\alpha T}]\\
&= \ee^{-\alpha t} \,\E\left[\prod_{i,j} \ee^{\alpha T_{i,j}}\right]\\
&= \ee^{-\alpha t} \,\E\left[\prod_{i,j} J_{i,j}!^\alpha \right].
\end{align*}
Of course these $J_{i,j}$'s are not independent of one another,
so bounding the expectation requires some work.
The above discussion motivates the following theorem:

\begin{theorem} \label{T-randomness-left}
Suppose that $X_1,X_2,\dots,X_\tau$ is a random-length sequence of random numbers
(where $\tau$ is a stopping time),
 such that
conditional on $t\leq \tau$,
 $X_t$ can be written
as $X_t=Y_t+U_t$ where $U_t$ is uniformly random in $[0,1)$ and independent of $X_1,\dots,X_{t-1}$ and $Y_t$.
Let
\[J_i=\#\{t\leq \tau: \lfloor X_t\rfloor=i\}\,,\]
denote the number of $X_t$'s that fall into ``bucket $i$'', and let
\[J^{(b)}_{i,j} = \#\{t\leq \tau:i+j/b\leq X_t<i+(j+1)/b\}\]
denote the number of these $X_t$'s
which are in the ``$j$th sub-bucket of the $i$th bucket'' if it is evenly split into $b$ sub-buckets.
If $\tau$ is bounded, then
 \[\E\left[5^{-\tau} \prod_i \max_{b_i\geq J_i} J^{(b_i)}_{i,0}!\cdots J^{(b_i)}_{i,b_i-1}!\right]\leq 1\,.\]
(If the $Y_t$'s are integer-valued, then $5^{-\tau}$ may be replaced with $3^{-\tau}$.)
\end{theorem}
\begin{proof}
  For $t\leq\tau$, let \[A_t=\#\left\{s\leq t: \lfloor X_s\rfloor = \lfloor X_t\rfloor\right\}\]
  be the number of points $X_1,\dots,X_t$ which are ``approximately'' $X_t$ ($X_t$ is approximately itself, so $A_t\geq 1$).
  Let \[C_t=\#\left\{s\leq t:\lfloor X_s\rfloor = \lfloor X_t\rfloor\ \text{and}\ |X_s-X_t|\leq \frac{1}{A_t}\right\}\]
  be the number of points $X_1,\dots,X_t$ to which $X_t$ is ``close''
  ($X_t$ is close to itself, so $C_t\geq 1$).
  Writing $X_t=Y_t+U_t$, we see that $X_t$ could be close to the items
  in bucket $\lfloor Y_t\rfloor$ or bucket $\lceil Y_t \rceil$ (which might be the same).
  Let $k$ denote the number of items in buckets $\lfloor Y_t\rfloor$ before item $t$ is added.
  For each item in bucket $\lfloor Y_t\rfloor$, the probability that $X_t$ is close to that item is at most $2/(k+1)$, so the expected number of
  such items is less than $2$.  Similarly, the expected number of items in bucket $\lceil Y_t\rceil$ to which $X_t$ is close is also less than $2$.
  Thus $\E[C_t\mid X_1,\dots,X_{t-1},Y_t]<5$, and if $Y_t$ is an integer then $\E[C_t\mid X_1,\dots,X_{t-1},Y_t]<3$.
  Since $\tau$ is bounded, by the Optional Stopping Theorem for supermartingales, we have $\E[5^{-\tau} C_1\cdots C_\tau]\leq 1$.
  For any $b_i\geq J_i$, if two items end up in the same sub-bucket then they must be close,
  so we have $J^{(b)}_{i,0}!\cdots J^{(b)}_{i,b_i-1}! \leq \prod_{t:\lfloor X_t\rfloor = i} C_t$,
  and hence $\prod_i J^{(b)}_{i,0}!\cdots J^{(b)}_{i,b_i-1}! \leq C_1 \cdots C_\tau$.
\end{proof}

For the priority queue application, the $X_t$'s correspond to the keys rescaled by a factor of~$W$,
the top-level bucket corresponds to $\lfloor X_t\rfloor$, $J_i$ corresponds to the number of items in top-level bucket~$i$,
the intervals $[i+j/b,i+(j+1)/b)$ correspond to the ranges of the sub-buckets, and $J^{(b_i)}_{i,j}$ corresponds to the number of items in the $j$th sub-bucket of top-level bucket $i$ if it is split into $b_i$ sub-buckets.  (The keys are approximately uniformly distributed within
a top-level bucket, assuming the top-level bucket width $W$ is not too large.
We will address the non-uniformity with the following lemmas.)  Even if an adversary
looks at the keys one by one and then decides when to stop adding new
items to the top level bucket, and even if the adversary is allowed to
choose the number of sub-buckets to be an arbitrary number which is at
least the number of items, still the expected product of the factorials of the
sub-bucket sizes cannot be too large.

\begin{lemma} \label{lem:P-almost-uniform}
  Suppose the forward-backward algorithm is run on $\G_n(\EXP(1))$.  The
  keys of the items inserted into queue $P$ can be reordered and
  written as $u_1,u_2,\dots,u_g,r_1,\dots,r_h$, where $g$ and $h$ are
  random, such that
\begin{enumerate}
\item conditional on $t\leq g$ and conditional on the values of $u_1,\dots,u_{t-1}$,
  there is a random integer $y_t$ for which $u_t/W$ is
  uniformly random on $[y_t,y_t+1)$,
\item each of the keys $u_1,\dots,u_g$ is in its top-level bucket when that bucket becomes active, and
\item $h$ is stochastically dominated by
    $\Poi(2 n^2 W)+2\Poi(n^2 W)+2\Poi(n N W)+\Poi(n^2W)+1+\Poi(3n^2W)$,
  where these different Poisson random variables may be correlated, and $N$ is the number of non-SPT edges enqueued in $P$.
  In particular, $\Pr[h>\Theta(n^2 W)] \leq \exp[-\Theta(n^2 W)]+\exp[-\Theta(n)]$.
\end{enumerate}
\end{lemma}
\begin{proof}
To do this partitioning of the edges in $P$, we can explore the
shortest path tree and the edge weights in a manner which respects the
top-level buckets as follows: Let \[S_i=\{u\in V:d_u<i W\}\] denote the of
vertices whose distance from the source would fall into one of the
first $i$ top-level buckets for $P$.  Once $S_i$ and the distances to
all vertices $u$ in $S_i$ have been determined, for each vertex $v$,
we determine whether or not $i W \leq d_u+c[u,v]<(i+1)W$, without
determining the fractional part of $d_u+c[u,v]$.
Not necessarily all such edges are in bucket~$i$ at the time that it
becomes active.  Let $G_i$ denote the set of edges $(u,v)$ which
satisfy the following properties:
\begin{enumerate}
\item $iW\leq d_u+c[u,v]<(i+1)W$,
\item $d_u<i W$,
\item $c[u,v]$ is the only point of edge $(u,v)$'s Poisson point process less than $(i+1)W-d_u$,
\item there is no $v'\neq v$ for which $iW\leq d_u+c[u,v']<(i+1)W$,
\item at least one of the following holds:
\begin{enumerate}
\item $|S_i|<\lceil n/2\rceil$,
\item $|S_i|\geq\lceil n/2\rceil$ and $(i+1)W\leq 2M-d_u$,
\item $|S_i|\geq\lceil n/2\rceil$ and $d_u\leq M$ and there is no $v'$ for which $M\leq d_u+c[u,v']<i W$, or
\item $|S_i|\geq\lceil n/2\rceil$
 and $(i+1)W\leq d_u+2\,(\min(d_v,i W)-M)$.
\end{enumerate}
\end{enumerate}
We will make the list $u_1,\dots,u_g$ by listing the keys of the edges in $\cup_i G_i$ so that edges in $G_i$ come before edges in $G_j$ when $i<j$, and the edges within $G_i$ appear in an arbitrary order (say lexicographic) which is independent of the edge weights.  There are several things to check to ensure that $u_1,\dots,u_g$ satisfies the claims of the lemma.

We first check claim 1 of the lemma:

Property 1 is that edge $(u,v)$'s key $d_u+c[u,v]$ falls into bucket~$i$.
Properties 1, 2, and~3 ensure that the key $d_u+c[u,v]$ is uniformly random in $[i W,(i+1)W)$.  Properties 4 and~5 can be tested knowing the graph only out to distance $i W$ from the source, together with knowledge of which keys fall into bucket~$i$; in particular, whether or not these properties hold is independent of the values of the keys falling into bucket~$i$.  Thus, conditional on the edges in $G_0,\dots,G_{i-1}$ and their keys, and on the edges in $G_i$, the keys of the edges in $G_i$ are i.i.d.\ uniform on $[i W,(i+1)W)$.
So $u_1,\dots,u_g$ satisfies claim 1 of the lemma.

To check claim 2 of the lemma, there are several cases depending on how an edge satisfies property 5:

If edge $(u,v)$ satisfies property 5a, then the algorithm is in phase 1 at the time that bucket~$i$ becomes active, so property 2 implies $P$ contains some outgoing edge of $u$ at this time.  Properties 1 and 4 imply that $(u,v)$ is the lightest outgoing edge of $u$ not yet extracted from $P$ when bucket~$i$ becomes active, so $P$ must contain edge $(u,v)$ at this time.

If edge $(u,v)$ satisfies property 5b, then $d_u+c[u,v]<(i+1)W\leq 2M-d_u$, so $c[u,v]\leq 2\,(M-d_u)$, so $(u,v)$ is out-pertinent.  Since $(u,v)$ is out-pertinent and $iW\leq d_u+c[u,v]$, queue~$P$ contains an out-pertinent edge of $u$ when bucket~$i$ becomes active; property 4 implies that $(u,v)$ is the lightest such edge with key at least $iW$, so $P$ must contain edge $(u,v)$ when bucket~$i$ becomes active.

If edge $(u,v)$ satisfies property 5c, then queue~$P$ contained an outgoing edge of $u$ at the time that the median $M$ was found, so its key is at least $M$, and this edge was not extracted before bucket~$i$ became active.  Properties 1, 4, and 5c together imply that $(u,v)$ is the lightest outgoing edge of $u$ with key $\geq M$, so $P$ must contain $(u,v)$ when bucket~$i$ becomes active.

If edge $(u,v)$ satisfies property 5d, then
$d_u+c[u,v]<(i+1)W\leq d_u+2\,(d_v-M)$, so the edge is in-pertinent.
At the time that bucket~$i$ becomes active, $\min(P)\geq i W$, so
$d_u+c[u,v]<(i+1)W\leq d_u+2\,(i W-M)\leq d_u+2\,(\min(P)-M)$, i.e.,
$c[u,v]< 2\,(\min(P)-M)$, so edge $(u,v)$ is enqueued in and extracted from $Q$ and requested by the time that bucket~$i$ becomes active.  If $(u,v)$ is not in $P$ at this time, it is because there is some other edge $(u,v')$ with smaller key in $P$.  However, property 4 ensures that any such edge is extracted from $P$ by the time that bucket~$i$ becomes active, so that $(u,v)$ is enqueued in $P$ by this time.

In each case, edge $(u,v)\in G_i$ is enqueued in $P$ by the time that bucket~$i$ becomes active.  Thus claim 1 of the lemma is true.

Next we check claim 3 of the lemma, i.e., that most edges enqueued in $P$ are among $u_1,\dots,u_g$.  There are several cases depending on which of the above properties that an edge enqueued in $P$ might fail:

For any edge $(u,v)$, property 1 holds for some $i$.

If property 2 fails, then $c[u,v]<W$.  The number of edges with such small cost is dominated by $\Poi(n^2 W)$.
Regardless of any of the edge weights, the number of edges which fail property 3 is dominated by $\Poi(n^2 W)$.
So the number of edges (whether or not they are enqueued in $P$) which fail property 2 or 3 is dominated by $\Poi(2n^2W)$.

Next we bound the number of SPT edges which fail property 4.
Recall the characterization of the distances of the vertices $0=d_{v_1}\leq d_{v_2}\leq\cdots\leq d_{v_n}$ from the source.  As we explore the graph $\G_n(\EXP(1))$ in order of increasing distance from the source, including both SPT edges and non-SPT edges, if an edge $(u,v)$ is found for which $d_u+c[u,v]$ is at most $W$ larger than the distance $d_{v'}$ to a vertex~$v'$ with parent~$u$, we can put a mark on vertex~$v'$ unless it already has a mark.  When a new vertex is discovered, the indicator random variable for it getting a mark is dominated by $\Poi(nW)$ regardless of the portion of the graph discovered so far or how many marks the previously discovered vertices get.  Since $n$ vertices are discovered, the total number of marks is dominated by $\Poi(n^2W)$.  If we double the number of these marks, we obtain an upper bound on the total number of SPT edges which fail property 4, together with some of the non-SPT edges which fail property 4.  In particular, if a bucket contains one or more SPT edges and zero or more non-SPT edges with the same parent $u$, we can cover all the SPT edges and the first (if any) of the non-SPT edges by these marks.

Next we bound the number of non-SPT edges enqueued in $P$ which fail property 4.
If we condition on the SPT and the distances to the vertices, then the remaining edge weights are independent of one another.
It is convenient to order the edges $(u,v)$ lexicographically with the endpoint $v$ being more significant, and then sample
the edge weights $c[u,v]$ in this order.
Once the edge weight $c[u,v]$ is sampled,
we can determine whether or not $(u,v)$ is out-pertinent or in-pertinent, but we may not be able to determine whether or not it is out-extraneous until subsequent edge weights are sampled.  If there was an earlier edge $(u,v')$ which was a candidate for being out-extraneous, and $(u,v)$ is a candidate for being out-extraneous, then $(u,v')$ is no longer a candidate for being out-extraneous.
For the $k$th vertex $v_k$, after all the non-SPT edges $(u,v_k)$ have had their weights sampled,
let $N_k$ denote the total number of non-SPT edges $(u,v_j)$ where $j\leq k$ which are pertinent or current candidates for being out-extraneous.
Let $X_k$ denote the number of these edges $(u,v)$ whose key falls into the same bucket as an earlier non-SPT edge $(u,v')$
that is enqueued in $P$.
Now $X_{k+1}-X_k$ is dominated by $\Poi(N_k W)$.  The $N_k$'s are non-decreasing, and $N_n$ is the total number of non-SPT edges enqueued in $P$.  From this we see that the number of non-SPT edges enqueued in $P$ which fail property 4 is dominated by
$2\Poi(n N_n W)$.

Next we bound the number of edges enqueued in $P$ which fail property 5.  There are several subcases depending on the type of edge enqueued in $P$:

Suppose that an edge $(u,v)$ is out-pertinent.  Then $c[u,v]\leq 2\,(M-d_u)$, so $iW\leq d_u+c[u,v]\leq 2M-d_u$.  If $(u,v)$ fails property 5a and 5b, then $|S_i|\geq\lceil n/2\rceil$ and $2M-d_u<(i+1)W$.  For any vertex $u$ for which $d_u\leq M$, there is only one bucket~$i$ for which $i W \leq 2M-d_u < (i+1)W$.
Regardless of the sequence of the first $\lceil n/2\rceil$ vertices whose distances are found, and the values of their distances, the number of out-pertinent edges which fail property 5 is stochastically dominated by a $\Poi(n^2 W)$ random variable.

Suppose that edge $(u,v)$ is out-extraneous.  Then $d_u\leq M$, and there is no vertex $v'$ for which $M<d_u+c[u,v']<d_u+c[u,v]$.
If $(u,v)$ is out-extraneous and fails property 5a and~5c, then $|S_i|\geq \lceil n/2\rceil$, $M<iW$, and there is some $v'$ for which $M\leq d_u+c[u,v']<iW\leq d_u+c[u,v]$.  Therefore $M=d_u+c[u,v']$.  With probability $1$ there is at most one such vertex $u$, and therefore at most one out-extraneous edge $(u,v)$ which fails property 5.

Suppose that $(u,v)$ is an in-pertinent edge.
Then $c[u,v]<2\,(d_v-M)\leq 2\,(d_u+c[u,v]-M)$ so $d_u+c[u,v]>2M-d_u$.
If 5a fails then $|S_i|\geq\lceil n/2\rceil$, so $d_u+c[u,v]>M$.
Suppose 5d also fails.
We now distinguish two subsubcases:

If $d_v\geq iW$, then
\begin{align*}
d_u+2\,(iW-M)
&<(i+1)W\\
iW &< 2M-d_u+W\\
2M-d_u<d_u+c[u,v] < iW+W &< 2M-d_u+2W\,.
\end{align*}
As the SPT and the edge weights are explored from the source, we see that for each vertex $u$,
the number of such exceptional edges (in-pertinent, $d_v\geq iW$, satisfies property 2, but fails property 5a and 5d) is dominated by a $\Poi(2nW)$ random variable which
is independent of the Poisson random variables associated with the other vertices.
So the total number of such exceptional edges is dominated by $\Poi(2n^2W)$.

If $d_v<iW$, then $(u,v)$ is not an SPT edge, and
\[
iW\leq d_u+c[u,v]<d_u+2\,(d_v-M) <(i+1)W\,.
\]
Given $d_u$, $d_v$, and $M$, there is only one choice of $i$ for which this could hold,
which imposes the following constraint on $c[u,v]$:
\[
\left\lfloor\frac{d_u+c[u,v]}{W}\right\rfloor =
\left\lfloor\frac{d_u+2\,(d_v-M)}{W}\right\rfloor\,.
\]
Even if we condition on the SPT and the distances to each vertex, and indeed, even
if we condition on all the other edge weights, for each non-SPT edge,
the indicator random variable for this event is dominated by $\Poi(W)$.
The total number of such exceptional edges (in-pertinent, $d_v<i W$, satisfies property 2, but fails property 5a and 5d) is dominated by a $\Poi(n^2 W)$ random variable, independent of the number of in-pertinent edges for which $d_v\geq iW$, satisfy property 2, but fail property 5.
Thus the total number of in-pertinent edges which satisfy property 2 but fail property 5 is at most $\Poi(3n^2W)$.
\end{proof}

\begin{lemma} \label{lem:Q-almost-uniform}
  Suppose the forward-backward algorithm is run on $\G_n(\EXP(1))$.  The
  keys of the items inserted into queue $Q$ can be reordered and
  written as $u_1,u_2,\dots,u_g,r_1,\dots,r_h$, where $g$ and $h$ are
  random, such that
\begin{enumerate}
\item conditional on $t\leq g$ and conditional on the values of $u_1,\dots,u_{t-1}$,
  there is a random number $y_t$ for which $u_t/W$ is
  uniformly random on $[y_t,y_t+1)$,
\item each of the keys $u_1,\dots,u_g$ is in its top-level bucket when that bucket becomes active, and
\item $h$ is stochastically dominated by
    $\Poi(16 n^2 W)+\Poi(4 n^2 W)+\Poi(8 n^2 W)+2\Poi(10 n^2 W)+2\Poi(10 n^2 W)+2\NB(\NB(\lfloor n/2\rfloor,1/4),5nW)+\Poi(n^2 W)+\Poi(n^2 W)+\Poi(n^2 W)+\Poi(n^2 W)$,
  where the above summands may be correlated, and $\NB$ denotes the negative binomial distribution.
  In particular, $\Pr[h>\Theta(n^2W)]\leq\exp[-\Theta(n^2 W)]$.
\end{enumerate}
\end{lemma}
\begin{proof}
The priority queue $Q$ is for processing the in-pertinent edges;
the algorithm will insert all in-pertinent edges and up to one additional incoming edge for each vertex, which we call the in-extraneous edges.
Recall that Theorem~\ref{T-PQ}(i) characterizes the in-extraneous edges.

We say that an edge $(u,v)$ is \textit{preactive\/} if it satisfies
\begin{enumerate}
\item \label{qQ-no-up} there is no $u'\neq u$ for which $\lfloor c[u,v]/W\rfloor = \lfloor c[u',v]/W\rfloor$, and
\item \label{qQ-in-Q} there is no $u'\neq u$ for which $2\,(d_v-M)\leq c[u',v] \leq c[u,v]$.
\end{enumerate}
Property~\ref{qQ-in-Q} ensures that edge $(u,v)$ is either in-pertinent ($c[u,v]<2\,(d_v-M)$) or else it is the unique lightest non-in-pertinent edge, in which case it is in-extraneous.  Either way, property~\ref{qQ-in-Q} implies that edge $(u,v)$ is enqueued in $Q$.  Property~\ref{qQ-no-up} ensures that vertex $v$ has no other incoming edge that might fall into the same bucket of $Q$ that $(u,v)$ does --- this ensures that $(u,v)$ is enqueued in $Q$ by the time that its top-level bucket becomes active.

Suppose we are given the graph structure of the SPT and the
values $\lfloor(d_u+c[u,v])/W\rfloor$ for each edge $(u,v)$, but
not the fractional part of $(d_u+c[u,v])/W$ for any edge $(u,v)$.
By induction $\lfloor d_v/W\rfloor$ is determined for each vertex $v$, so
$\lfloor M/W\rfloor$ is determined.
Consequently, each edge cost $c[u,v]$ is determined to within an open interval of width $2W$.
Given this partial information, it may be that no matter what the fractional parts
of the various edge costs are, edge $(u,v)$ is guaranteed to satisfy properties 1 and/or 2,
in which case we say that $(u,v)$ \textit{robustly\/} satisfies these properties.
We say that $(u,v)$ is \textit{robustly preactive\/} if it robustly satisfies both properties.
We now argue that most edges enqueued in $Q$ are robustly preactive.

If an in-pertinent SPT edge $(u,v)$ does not robustly satisfy
property~2, then $|2\,(d_v-M)-c[u,v]|\leq 4W$, and since
$d_v=d_u+c[u,v]$, it follows that $|2\,(M-d_u)-c[u,v]|\leq 4W$.
Recalling the SPT exploration process, we see that for each
vertex~$u$, the number of edges $(u,v)$ satisfying this constraint is
dominated by a $\Poi(8 n W)$ random variable $Y_u$, where the $Y_u$'s
are independent.  So the number of in-pertinent SPT edges which do not
robustly satisfy property~2 is dominated by $\Poi(8 n^2 W)$.

To analyze the other types of edges that may be enqueued in $Q$, we
condition on the shortest path tree and the values of the
distances to all the vertices, and take the cost of each non-SPT edge
$(u,v)$ to be $c[u,v]=d[v]-d[u]+\EXP(1)$ independent of the other edge
costs.

If $(u,v)$ is an in-pertinent non-SPT edge, then the indicator random
variable for the event that $|2\,(d_v-M)-c[u,v]|\leq 4W$ is dominated by
$\Poi(8W)$, independent of the other edge costs.  So the number of
non-SPT in-pertinent edges which do not robustly satisfy property~2 is
also dominated by $\Poi(8 n^2 W)$.  Because of the independence, the
total number of in-pertinent edges which do not robustly satisfy
property~2 is dominated by $\Poi(16 n^2 W)$.

If an in-extraneous edge $(u,v)$ does not robustly satisfy
property~2, it must be that $d_v\geq M$, $2\,(d_v-M)\leq c[u,v]$, there is
no $u'$ for which $2\,(d_v-M)\leq c[u',v]<c[u,v]$ (otherwise it would
not be in-extraneous), and there is some $u'\neq u$ for which
$c[u',v]\leq c[u,v]+4 W$.
We distinguish two cases depending on whether the
SPT edge leading to $v$ is in-pertinent or out-pertinent.  The edge
costs of those edges leading to $v$ but heavier than $2\,(d_v-M)$
is dominated by a Poisson process (in the in-pertinent case), or a Poisson
process with one extra point (in the out-pertinent case).  Whatever the
first point of the Poisson process, the number of additional points which
are at most $4W$ larger is dominated by $\Poi(4 n W)$.  In the out-pertinent
case, there are at most $\Poi(8 n W)$ Poisson process points within distance~$4W$
of the extra point.  Either way, the number of in-extraneous edges leading
to $v$ which do not robustly satisfy property~2 is dominated by a $\Poi(4 n W)+\Poi(8 n W)$
random variable~$Y_v$, where the $Y_v$'s are independent.  Thus there are at most
$\Poi(4 n^2 W)+\Poi(8 n^2 W)$ in-extraneous edges which fail to robustly satisfy property~2.

Next we bound the number of edges enqueued in $Q$ which do not
robustly satisfy property~1.

Suppose we are given just the SPT and the values
$\lfloor(d_x+c[x,y])/W\rfloor$ for each edge $(x,y)$.  If
$|c[u,v]-c[u',v]|\geq 5W$, then these edge costs are sufficiently far
apart that we can deduce that $\lfloor c[u,v]/W\rfloor \neq \lfloor
c[u',v]/W\rfloor$.  We say that two edge costs $c[u,v]$ and $c[u',v]$
are a ``close encounter'' if they differ by less than $5W$, and there
is no vertex $u''$ for which $c[u'',v]$ lies between $c[u,v]$ and
$c[u',v]$.  If edge $(u,v)$ is not in a close encounter, then it
robustly satisfies property~1, and each close encounter causes out at most
two edges to fail to robustly satisfy property~1.

As above, we condition on the SPT and the distances.

If $(u,v)$ is an SPT edge, then the number of edges $(u',v)$ for which
$u'\neq u$ and $|c[u',v]-c[u,v]|\leq 5W$ is dominated by $\Poi(10nW)$,
independent of the other SPT edges, so the number of SPT edges
enqueued in $Q$ which fail to robustly satisfy property~1 is dominated
by $\Poi(10 n^2 W)$.

Conditional on the SPT and distances, the costs of the non-SPT edges
leading to $v$ is dominated by a Poisson process independent of the
other vertices.  The number of the Poisson
points within distance $5W$ of the cost of the SPT edge leading to $v$
is dominated by $\Poi(10 n W)$.  The number of these points that are
within distance $5W$ of the first point beyond $2\,(d_v-M)$ is also at
most $\Poi(10 n W)$.  Next we bound the close encounters between the
Poisson points up to distance $2\,(d_v-M)$.  Scanning the Poisson
process left-to-right starting from $0$, when a point is encountered,
the scan becomes ``active'', but becomes ``inactive'' if no additional
point is encountered with distance $5W$.  While the scan is active,
each additional point scanned becomes a close encounter, and resets
the clock during which the scan is active.  The number of times that
the scan becomes active is dominated by $\Poi(n\times 2\,(d_v-M))$, and
for each time that the scan is active, the number of close encounters
is dominated by a geometric random variable (starting at $0$) with success probability $e^{-5W}$, where
all the Poisson and geometric random variables are independent of one
another.  Recall that the ensemble of random variables $\{d_v-M\}_{v\in V}$ is
dominated by a collection of $\lfloor n/2\rfloor$ i.i.d.\ $\EXP(1)/\lceil n/2\rceil$ random
variables, and that $\Poi(\alpha\EXP(1))$ is a geometric random
variable (with minimum value 0) with mean $\alpha$.  So the number of
non-SPT edges enqueued in $Q$ which do not robustly satisfy property~1
is dominated by $\lfloor n/2\rfloor$ geometrics of geometrics, and is unlikely to be
much larger than $\Theta(n^2 W)$.

Next we argue that the edges enqueued in $Q$ which are robustly preactive have approximately random costs modulo $W$, in a sense which we now make precise.

For each edge $(u,v)$, we can take its cost to be $-d_u$ plus the first point larger than $d_u$ of a Poisson point process on $[0,\infty)$.  The edge cost is not known until $d_u$ is known, but if we explore the SPT tree from the root, we do not need to know $(u,v)$'s cost until $d_u$ is known.  Using this realization of the edge costs, we say that the edge $(u,v)$ is \textit{flexible\/} if
\\[6pt]
\noindent\hspace*{12pt} Letting $i=\lfloor (d_u+c[u,v])/W\rfloor$, we have
\begin{enumerate}
\item $d_u<i W$,
\item $d_u+c[u,v]$ is the only point of edge $(u,v)$'s Poisson point process in the interval $[iW,(i+1)W)$, and
\item Either
\begin{enumerate}
\item
 $d_v < i W$, or
\item there is no vertex $w$ for which edge $(v,w)$'s Poisson point process has a point in the interval $[i W,(i+1)W)$.
\end{enumerate}
\end{enumerate}
Conditional upon an edge $(u,v)$ being flexible, the value of $d_u+c[u,v]$ is uniformly distributed in the interval
$[i W,(i+1)W)$.
Furthermore, changing the cost of a flexible edge $(u,v)$ while preserving
 $\lfloor (d_u+c[u,v])/W\rfloor$
 has no effect on either the graph structure of the SPT or whether or not other edges are flexible.

Edges which fail property~1 of being flexible have cost at most $W$, so the number of them is dominated by $\Poi(n^2 W)$.  The number of edges which fail property~2 is also dominated by $\Poi(n^2 W)$.  SPT edges of course fail property~3a, but the number of them that fail property~3b is dominated by $\Poi(n^2 W)$.
Some non-SPT edges are also enqueued in~$Q$.  Any edge $(u,v)$ which fails property~3a satisfies $d_v\leq d_u+c[u,v]<d_v+W$.  Conditional on the SPT and distances to each vertex, the number of non-SPT edges which fail property~3a is dominated by $\Poi(n^2 W)$, so this is also true unconditionally.

The keys of the edges which are both robustly preactive and flexible will be the keys $u_1,\dots,u_g$.  This set of keys satisfy properties 1 and 2 of the lemma.  When we tally up the edges enqueued in $Q$ which fail to be robustly preactive or flexible, so total is small, which is property 3 of the lemma.
\end{proof}

\begin{theorem}
  Suppose the forward-backward algorithm's priority queues $P$ and $Q$ are
  implemented as two-level bucket monotone priority queues with
  sub-buckets based on binary heaps, as described Section~\ref{sub:two-level}.
  If $W=1/(n\log n)$ and $B=n$,
  then when the algorithm is run on
  $\G_n(\EXP(1))$, the expected running time for priority queues $P$ and $Q$ is
  $\Theta(n)$, and the running time is $\Theta(n)$ except with probability $\exp(-\Theta(n/\log n))$.
\end{theorem}

\begin{proof}
Janson \cite[eq.~2.8]{Janson99} proved that for $\G_n(\EXP(1))$,
\[ \Pr\left[\max_v d_v \geq 2\frac{\log n + \log \log n}{n} + \frac{t}{n}\right] \leq \exp[-t(1-1/\log n) + O(1)]\,. \]
With $B=n$ and $W=1/(n\log n)$, the relevant $t$ is $\Theta(n/\log n)$, and we see that the probability that the last bucket becomes active is at most $\exp(-\Theta(n/\log n))$.
Of course if no items are extracted from the last bucket, then we only need to pay for the edge insertion costs into the last bucket's linked list.

Observe that
\begin{align*}
\binom{a_1+\cdots+a_\ell}{a_1,\ldots,a_\ell} &\leq \ell^{a_1+\cdots+a_\ell} \\
(a_1+\cdots+a_\ell)! &\leq a_1!\times\cdots\times a_\ell! \times \ell^{a_1+\cdots+a_\ell}\,.
\end{align*}
\newcommand{\Rmax}{R_{\max}}
This inequality allows us to handle the ``good keys'', i.e., those which are in their top-level bucket by the time it becomes active and whose value modulo $W$ is random (those keys satisfying conditions~1 and~2 in Lemmas~\ref{lem:P-almost-uniform} and~\ref{lem:Q-almost-uniform}), separately from the small number of other keys.  We let $J_i$ denote the number of good keys that fall into bucket~$i$, let $b_i\geq J_i$ denote the number of keys in bucket~$i$ at the time that it becomes active, and let $R_i$ denote the number of remaining keys inserted into bucket~$i$ which are not good in the above sense.  Let $J_{i,j}$ (resp.\ $R_{i,j}$) denote the number of good (resp.\ bad) keys that fall into sub-bucket~$j$ of bucket~$i$.  Let $J=\sum_i J_i$ and $R=\sum_i R_i$, and let $\Rmax=\Theta(n/\log n)$ be a value that we will discuss later.
We have
\begin{align*}
\prod_{j=0}^{b_i-1} (J_{i,j}+R_{i,j})!
 &\leq 2^{J_i+R_i}\times \prod_{j=0}^{b_i-1}J_{i,j}! \times \prod_{j=0}^{b_i-1}R_{i,j}! \\
\prod_i\prod_{j=0}^{b_i-1} (J_{i,j}+R_{i,j})! \times 1_{R\leq \Rmax}
 &\leq 2^{J+R}\times \prod_i \prod_{j=0}^{b_i-1}J_{i,j}! \times R!\times 1_{R\leq \Rmax}\\
\E\left[10^{-J} 2^{-R}\prod_i\prod_{j=0}^{b_i-1} (J_{i,j}+R_{i,j})!\times 1_{R\leq \Rmax}\right]
 &\leq \E\left[5^{-J}\times \prod_i \prod_{j=0}^{b_i-1}J_{i,j}! \times \Rmax!\right] \\
 &= \E\left[5^{-J}\times \prod_i \prod_{j=0}^{b_i-1}J_{i,j}!\right] \times \Rmax! \\
 &\leq \Rmax!\,.
\end{align*}

Recall now that the time required to perform the \extractmin\ binary heap operations for sub-bucket~$j$ of bucket~$i$ is proportional to $(J_{i,j}+R_{i,j})+\log (J_{i,j}+R_{i,j})!$.  Equivalently, the time is proportional to $(1+\log 10)J_{i,j}+(1+\log 2)R_{i,j} + Y_{i,j}$ where
\[Y_{i,j} = - J_{i,j} \log 10 - R_{i,j} \log 2 + \log (J_{i,j}+R_{i,j})!\,.\]
The total time for the priority queue operations is proportional to $\text{const}\times(J+R)+Y$, where
\[Y=\sum_i \sum_{j=0}^{b_i-1} Y_{i,j}\,.\]
We have already proved in Theorem~\ref{T-pertinent-exp} that, with
suitable constants, the number of edges $J+R$ enqueued in the queues
is $O(n)$ except with probability $\exp[-\Theta(n)]$.
We are left to bound the probability that $Y$ is large:
\begin{align*}
\Pr[Y\geq t]
 &\leq \Pr\big[\ee^Y \times 1_{R\leq\Rmax} \geq \ee^t\big] + \Pr[R>\Rmax]\\
 &\leq \ee^{-t} \,\E[\ee^Y \times 1_{R\leq\Rmax}] + \Pr[R>\Rmax]\\
 &\leq \ee^{-t} \,\E\left[10^{-J} 2^{-R} \prod_i\prod_{j=0}^{b_i-1} (J_{i,j}+R_{i,j})! \times 1_{R\leq\Rmax}\right]  + \Pr[R>\Rmax]\\
 &\leq \ee^{-t} \,\Rmax!  + \Pr[R>\Rmax]\,.
\end{align*}
By Lemma~\ref{lem:P-almost-uniform} (for queue~$P$) and~\ref{lem:Q-almost-uniform} (for queue~$Q$), we can take $\Rmax=\Theta(n/\log n)$ and have $\Pr[R>\Rmax]\leq\exp[-\Theta(n/\log n)]$.  Then $\Rmax!=\exp[\Theta(n)]$, so we can take $t$ be a sufficiently large constant times $n$ to find $\Pr[Y\geq t]\leq\exp[-\Theta(n/\log n)]$.

In the exponentially unlikely event that the running time exceeds
$O(n)$, the conditional expected running time can at worst be
$O(n^2\log n)$, so this unlikely event contributes negligibly to the
total expected running time.
\end{proof}

\section{Lower bound for forward-only algorithms}\label{S-lower}

In this section we show that any algorithm which is given only the
outgoing edges from each vertex in sorted order must take at least
$\Omega(n\log n)$ expected time to produce the shortest path tree for
the directed complete graph with random edge weights.  We actually prove a
stronger lower bound, namely we show that even if an algorithm is
presented the shortest path tree, just the task of verifying that the
shortest path tree is correct requires the examination of
$\Omega(n\log n)$ edges on average.  This lower bound holds even for
\textit{non-deterministic\/} verification procedures, i.e.,
verification procedures that are allowed to guess which queries to
ask.  These lower bounds complement our algorithmic results, by
showing that having both outgoing edges and incoming edges in sorted
order genuinely reduces the average-case complexity of finding the
shortest path tree.

The verification algorithms we consider are allowed to make two types of queries:
\begin{enumerate}
\item $query(u,v)$ -- What is the weight of the edge $(u,v)$?
\item $query_{F}(u,i)$ -- What is the $i$-th \emph{outgoing\/} edge of~$u$, according to non-decreasing order of weight, and what is its weight?
\end{enumerate}
We refer to these queries as \emph{edge queries\/} and \emph{forward queries}, respectively. Edge queries can be answered in constant time if a weighted adjacency matrix of the graph is available. Forward queries can be answered in constant time if the outgoing adjacency lists of each vertex are sorted in non-decreasing order of weight and are stored in arrays.

We say that an algorithm \emph{probed\/} an edge $(u,v)$ if it issued a direct $query(u,v)$, or if the edge $(u,v)$ and its weight were returned by $query_{F}(u,i)$, for some~$i$.

\begin{lemma} \label{lem:query}
 Any verification algorithm that only uses edge queries $query(u,v)$ must query all edges $(u,v)$ for which $d[u]<d[v]$ before accepting a SPT. In particular, if all distances are distinct, any such algorithm must query at least $\binom{n}{2}=\Omega(n^2)$ edges.
\end{lemma}

\begin{proof}
If the algorithm fails to query an edge $(u,v)$ for which $d[u]<d[v]$,
then it is possible that $c[u,v]<d[v]-d[u]$, in which case $T$ is not a SPT.
\end{proof}

Suppose that $T$ is a SPT corresponding to the cost function $c:E\to
\reals^+$.  To show that a given set of queries is not a sufficient to
conclude that~$T$ is indeed a SPT, we need to provide an alternative
cost function $c':E\to\reals^+$ under which the answers to all queries
made by the algorithm are the same, yet~$T$ is not a SPT under~$c'$.
In the proof of Lemma~\ref{lem:query}, the alternative cost function~$c'$
is the same as the cost function~$c$ except at some edge $(u,v)$ for
which $d[u]<d[v]$ and which the algorithm failed to query.

Next we consider verification algorithms which are allowed to make
forward queries.  Recall that our lower bounds are in fact for
\emph{non-deterministic\/} verification algorithms, i.e., algorithms
that are allowed to guess which queries to ask.  If the verification
procedure makes both edge queries and forward queries, then we can
replace each edge query $query(u,v)$ by a forward query
$query_F(u,i)$, where $i$ is the index of~$v$ in the sorted outgoing
adjacency list of~$u$, as the non-deterministic algorithm may guess
this index.  This clearly gives the verification algorithm more
information.

Let $v_1,\ldots,v_n$ be the vertices in non-decreasing order of
distance from~$s$ in~$T$, i.e., $0=d[v_1]\le d[v_2]\le \cdots \le
d[v_n]$, where $v_1=s$. We let $v_{i,j}$ be the endpoint of the $j$-th
outgoing edge of~$v_i$, in sorted order of weight, i.e.,
$c[v_i,v_{i,1}]\le c[v_i,v_{i,2}]\le\cdots\le c[v_i,v_{i,n-1}]$.
In other words, the answer to $query_F(v_i,j)$ is the pair $(v_{i,j},c[v_i,v_{i,j}])$.

\begin{lemma} \label{lem:one-of-two}
  Suppose that $T$ is a SPT with respect to a directed cost function $c:E\to\reals^+$.
  Let $v_i\in V$.  If $c[v_i,v_{i,r-1}]<d[v_{i,s}]-d[v_i]$, then any verification
  algorithm that only uses forward queries must either
  $query_F(v_i,r)$ or $query_F(v_i,s)$ before accepting~$T$.
\end{lemma}

\begin{proof}
 If $s\le r-1$, then $c[v_i,v_{i,s}]\le c[v_i,v_{r-1}] <
 d[v_{i,s}]-d[v_i]$, in which case $T$ would not be a shortest path
 tree.  We may assume then that $r\leq s$.  We define a new cost
 function $c':E\to\reals^+$ as follows.  For every edge $(u',v')$
 other than $(v_i,v_{i,r})$ or $(v_i,v_{i,s})$, we let
 $c'[u',v']=c[u',v']$.  Let $c'[v_i,v_{i,s}]=c[v_i,v_{r-1}]$ and
 $c'[v_i,v_r]=c[v_i,v_{i,s}]$. This amounts to swapping the roles of
 $v_{i,r}$ and $v_{i,s}$ followed by setting
 $c'[v_i,v_{i,s}]=c'[v_i,v_{i,r-1}]<d[v_{i,s}]-d[v_i]$.  (Observe
 that this construction of $c'$ would not work in the undirected
 setting.)  The shortest path tree for the cost functions $c$ and
 $c'$ are different, and the costs are the same except on the edges
 $(v_i,v_{i,r})$ and $(v_i,v_{i,s})$, so any verification algorithm
 must query at least one of these two edges.
\end{proof}

In the above proof, we did not have to let $c'[v_i,v_s]=c[u,v_{r-1}]$.  Any choice that satisfies $c'[v_i,v_{i,r-1}]\le c'[v_i,v_s] \le c'[u,v_{r+1}]$ and $c'[v_i,v_{i,s}]<d[v_{i,s}]-d[v_i]$ would do.  Thus, the lower bound holds even if the verification procedure is guaranteed that all edge weights are distinct.

\begin{theorem} \label{thm:lower-bound}
  Any verification algorithm that uses only edge and forward queries
  must perform $(1+o(1)) n \log n$ queries, with high probability,
  when given a SPT of the directed graph $\G_n(\EXP(1))$.
\end{theorem}

\begin{proof}
  For $1\leq i\leq n$ and $1\leq k<n$, we let $r_{i,k}$ be the minimal
  index for which \[c[v_i,v_{i,r_{i,k}}]\ge d[v_{n-k}]-d[v_i]\,.\] If
  there is some value of $r\leq r_{i,k}$ for which $query_F(v_i,r)$ is
  not queried, and some value of $s$ for which $d[v_{i,s}]\geq
  d[v_{n-k}]$ and $query_F(v_i,s)$ is not queried, then by
  Lemma~\ref{lem:one-of-two} a verification procedure cannot be
  certain that it has the correct SPT.  So a verification procedure
  must either $query_F(v_i,r)$ for each $r\leq r_{i,k}$ or $query_F(v_i,s)$
  for each $s$ for which $d[v_{i,s}]\geq d[v_{n-k}]$ (or both).
  Thus there must be at least $\min(r_{i,k},k)$ forward queries from vertex $v_i$.

  By the analyses of Davis and Prieditis \cite{DaPr93} and Janson
  \cite{Janson99}, we know that with high probability
  $d[v_n]=(2+o(1))(\log n)/n$, and that with high probability, for
  almost all vertices~$v_i$, $d[v_i]=(1+o(1))(\log n)/n$.  At this
  point we set $k=\lceil\log n\rceil$.  By the characterization of the
  distances, we have $\E[d[v_n]-d[v_{n-k}]]=(1+o(1))(\log\log n)/n$,
  so $d[v_{n-k}]=(2+o(1))(\log n)/n$ with high probability.

  For any vertex $u$, the number of outgoing edges $(u,w)$ with weight
  at most $(1+o(1))(\log n)/n$ is a binomial random variable with mean
  $(1+o(1))\log n$, so the actual number will be concentrated about
  the mean.  For a suitable choice of the $o(1)$ term, with high
  probability this will be less than $d[v_{n-k}]-d[u]$.  Thus, with high probability,
  for almost all vertices~$v_i$, we have $\min(r_{i,k},k)=(1+o(1))\log n$, and so
  the total number of forward queries is with high probability at least $(1+o(1)) n \log n$.
\end{proof}

Observe that we were able to make these alternative cost functions $c'$ for verification algorithms that only use edge and forward queries
since, in some sense,
there is no interaction between the forward queries made on different vertices.
When both forward-queries and backward-queries are allowed, it becomes harder to make on alternative cost function $c'$ for which the verification algorithm fails, which should not be surprising in view of the fact that we have shown that on average $O(n)$ forward-queries and backward-queries suffice to find the shortest-path tree.

\section{Concluding remarks and open problems}\label{S-concl}

We presented a new forward-backward single-source shortest paths algorithm that demonstrates the usefulness of backward scans. We used our algorithm to solve the SSSP problem in the complete graph with exponential edges weights in $O(n)$ time, with very high probability, which is clearly optimal.
As mentioned, we hope that ideas from our forward-backward algorithm may help speed up shortest paths algorithms not only in theory but also in practice.

We presented a probabilistic analysis of the forward-backward algorithm for exponential edge weights, as well as for Weibull edges weights, i.e., edge weights of the form $\EXP(1)^s$, when $0<s\le 1$. We conjecture that the algorithm runs in $O(n)$ time also when $s>1$.

Finally, an interesting open problem is whether backward scans may also yield improved SSSP or APSP algorithms in the more general \emph{end-point independent\/} model.

\makeatletter
\def\@rst #1 #2other{#1}
\newcommand\MR[1]{\relax\ifhmode\unskip\spacefactor3000 \space\fi
  \MRhref{\expandafter\@rst #1 other}{#1}}
\newcommand\MRhref[2]{\href{http://www.ams.org.offcampus.lib.washington.edu/mathscinet-getitem?mr=#1}{MR#1}}
\makeatother

\phantomsection
\pdfbookmark[1]{References}{bib}
\bibliographystyle{hmralphaabbrv}
\bibliography{mr}

\newcommand{\etalchar}[1]{$^{#1}$}
\begin{thebibliography}{AMOT90}

\bibitem[AMOT90]{AhMeOrTa90}
R.~K. Ahuja, K.~Mehlhorn, J.~B. Orlin, and R.~E. Tarjan.
\newblock Faster algorithms for the shortest path problem.
\newblock {\em J. Assoc.\ Comput.\ Mach.},
  \href{http://dx.doi.org/10.1145/77600.77615}{37(2):213--223}, 1990.
  \MR{1072256 (92e:68063)}

\bibitem[Blo83]{Bloniarz83}
P.~A. Bloniarz.
\newblock A shortest-path algorithm with expected time {$O(n^{2}{\rm log}\,n$}
  {${\rm log}^{\ast} n)$}.
\newblock {\em SIAM J. Comput.},
  \href{http://dx.doi.org/10.1137/0212039}{12(3):588--600}, 1983. \MR{707415
  (84i:68060)}

\bibitem[BPT{\etalchar{+}}73]{MR0329916}
M.~Blum, V.~Pratt, R.~E. Tarjan, R.~W. Floyd, and R.~L. Rivest.
\newblock Time bounds for selection.
\newblock {\em J. Comput.\ System Sci.},
  \href{http://dx.doi.org/10.1016/S0022-0000(73)80033-9}{7:448--461}, 1973.
\newblock Fourth Annual ACM Symposium on the Theory of Computing. \MR{0329916
  (48 \#8256)}

\bibitem[BvdH12]{MR2932542}
S.~Bhamidi and R.~van~der Hofstad.
\newblock Weak disorder asymptotics in the stochastic mean-field model of
  distance.
\newblock {\em Ann.\ Appl.\ Probab.},
  \href{http://dx.doi.org/10.1214/10-AAP753}{22(1):29--69}, 2012.
\newblock \arXiv{1002.4362}. \MR{2932542}

\bibitem[CGS99]{ChGoSi99}
B.~V. Cherkassky, A.~V. Goldberg, and C.~Silverstein.
\newblock Buckets, heaps, lists, and monotone priority queues.
\newblock {\em SIAM J. Comput.},
  \href{http://dx.doi.org/10.1137/S0097539796313490}{28(4):1326--1346}, 1999.
  \MR{1681014 (2000a:68047)}

\bibitem[Dan60]{Dantzig60}
G.~B. Dantzig.
\newblock On the shortest route through a network.
\newblock {\em Management Sci.},
  \href{http://dx.doi.org/10.1287/mnsc.6.2.187}{6:187--190}, 1959/1960.
  \MR{0113815 (22 \#4647)}

\bibitem[DI04]{DeIt04a}
C.~Demetrescu and G.~F. Italiano.
\newblock A new approach to dynamic all pairs shortest paths.
\newblock {\em J. ACM},
  \href{http://dx.doi.org/10.1145/1039488.1039492}{51(6):968--992}, 2004.
  \MR{2145260 (2006a:68133)}

\bibitem[DI06]{DeIt06}
C.~Demetrescu and G.~F. Italiano.
\newblock Experimental analysis of dynamic all pairs shortest path algorithms.
\newblock {\em ACM Trans.\ Algorithms},
  \href{http://dx.doi.org/10.1145/1198513.1198519}{2(4):578--601}, 2006.
  \MR{2284246 (2007j:68187)}

\bibitem[Dia69]{Dial69}
R.~B. Dial.
\newblock Algorithm 360: {S}hortest-path forest with topological ordering.
\newblock {\em Comm.\ ACM},
  \href{http://dx.doi.org/10.1145/363269.363610}{12(11):632--633}, 1969.

\bibitem[Dij59]{Di59}
E.~W. Dijkstra.
\newblock A note on two problems in connexion with graphs.
\newblock {\em Numer.\ Math.},
  \href{http://dx.doi.org/10.1007/BF01386390}{1:269--271}, 1959. \MR{0107609
  (21 \#6334)}

\bibitem[DP93]{DaPr93}
R.~Davis and A.~Prieditis.
\newblock The expected length of a shortest path.
\newblock {\em Inform.\ Process.\ Lett.},
  \href{http://dx.doi.org/10.1016/0020-0190(93)90059-I}{46(3):135--141}, 1993.
  \MR{1229199 (94e:90037)}

\bibitem[FG85]{FrGr85}
A.~M. Frieze and G.~R. Grimmett.
\newblock The shortest-path problem for graphs with random arc-lengths.
\newblock {\em Discrete Appl.\ Math.},
  \href{http://dx.doi.org/10.1016/0166-218X(85)90059-9}{10(1):57--77}, 1985.
  \MR{770869 (86g:05084)}

\bibitem[FKG71]{FKG}
C.~M. Fortuin, P.~W. Kasteleyn, and J.~Ginibre.
\newblock Correlation inequalities on some partially ordered sets.
\newblock {\em Comm. Math. Phys.}, 22:89--103, 1971. \MR{0309498 (46 \#8607)}

\bibitem[FT87]{FrTa87}
M.~L. Fredman and R.~E. Tarjan.
\newblock Fibonacci heaps and their uses in improved network optimization
  algorithms.
\newblock {\em J. Assoc.\ Comput.\ Mach.},
  \href{http://dx.doi.org/10.1145/28869.28874}{34(3):596--615}, 1987.
  \MR{904195 (90d:68012)}

\bibitem[Gol08]{Goldberg08}
A.~V. Goldberg.
\newblock A practical shortest path algorithm with linear expected time.
\newblock {\em SIAM J. Comput.},
  \href{http://dx.doi.org/10.1137/070698774}{37(5):1637--1655}, 2008.
  \MR{2386284 (2008m:68249)}

\bibitem[Hag06]{Hagerup06}
T.~Hagerup.
\newblock Simpler computation of single-source shortest paths in linear average
  time.
\newblock {\em Theory Comput.\ Syst.},
  \href{http://dx.doi.org/10.1007/s00224-005-1260-0}{39(1):113--120}, 2006.
  \MR{2189802 (2006h:68180)}

\bibitem[Hoa61]{find}
C.~A.~R. Hoare.
\newblock Algorithm 65: {Find}.
\newblock {\em Comm.\ ACM},
  \href{http://dx.doi.org/10.1145/366622.366647}{4(7):321--322}, 1961.

\bibitem[HZ85]{HaZe85}
R.~Hassin and E.~Zemel.
\newblock On shortest paths in graphs with random weights.
\newblock {\em Math.\ Oper.\ Res.},
  \href{http://dx.doi.org/10.1287/moor.10.4.557}{10(4):557--564}, 1985.
  \MR{812814 (87a:05096)}

\bibitem[Jan99]{Janson99}
S.~Janson.
\newblock One, two and three times {$\log n/n$} for paths in a complete graph
  with random weights.
\newblock {\em Combin.\ Probab.\ Comput.},
  \href{http://dx.doi.org/10.1017/S0963548399003892}{8(4):347--361}, 1999.
  \MR{1723648 (2000j:05113)}

\bibitem[LR89]{LuRa89}
M.~Luby and P.~Ragde.
\newblock A bidirectional shortest-path algorithm with good average-case
  behavior.
\newblock {\em Algorithmica},
  \href{http://dx.doi.org/10.1007/BF01553908}{4(4):551--567}, 1989. \MR{1019393
  (91e:68074)}

\bibitem[Mey03]{Meyer03}
U.~Meyer.
\newblock Average-case complexity of single-source shortest-paths algorithms:
  lower and upper bounds.
\newblock {\em J. Algorithms},
  \href{http://dx.doi.org/10.1016/S0196-6774(03)00046-4}{48(1):91--134}, 2003.
  \MR{2006098 (2004h:68105)}

\bibitem[MP97]{MePr97}
K.~Mehlhorn and V.~Priebe.
\newblock On the all-pairs shortest-path algorithm of {M}offat and {T}akaoka.
\newblock {\em Random Structures Algorithms},
  \href{http://dx.doi.org/10.1002/(SICI)1098-2418(199701/03)10:1/2<205::AID-RSA11>3.3.CO;2-J}{10(1-2):205--220},
  1997.
\newblock Average-case analysis of algorithms (Dagstuhl, 1995). \MR{1611525
  (2000e:68131)}

\bibitem[MT87]{MoTa87}
A.~Moffat and T.~Takaoka.
\newblock An all pairs shortest path algorithm with expected time {$O(n^2\,{\rm
  log}\,n)$}.
\newblock {\em SIAM J. Comput.},
  \href{http://dx.doi.org/10.1137/0216065}{16(6):1023--1031}, 1987. \MR{917038
  (89c:68050)}

\bibitem[Nic66]{Nicholson66}
T.~A.~J. Nicholson.
\newblock Finding the shortest route between two points in a network.
\newblock {\em The Computer Journal},
  \href{http://dx.doi.org/10.1093/comjnl/9.3.275}{9(3):275--280}, 1966.

\bibitem[Pet04]{Pettie04}
S.~Pettie.
\newblock A new approach to all-pairs shortest paths on real-weighted graphs.
\newblock {\em Theoret.\ Comput.\ Sci.},
  \href{http://dx.doi.org/10.1016/S0304-3975(03)00402-X}{312(1):47--74}, 2004.
  \MR{2045485 (2005b:05210)}

\bibitem[Poh71]{Pohl71}
I.~Pohl.
\newblock Bi-directional search.
\newblock In {\em Machine {I}ntelligence 6}, pages 127--140. Edinburgh
  University Press, 1971. \MR{0274003 (42 \#8879)}

\bibitem[PSSZ13]{PSSZ}
Y.~Peres, D.~Sotnikov, B.~Sudakov, and U.~Zwick.
\newblock All-pairs shortest paths in {$O(n^2)$} time with high probability.
\newblock {\em J. ACM},
  \href{http://dx.doi.org/10.1145/2508028.2505988}{60(4):Art. 26, 25}, 2013.
  \MR{3116511}

\bibitem[Spi73]{Spira73}
P.~M. Spira.
\newblock A new algorithm for finding all shortest paths in a graph of positive
  arcs in average time {$O(n^{2}{\rm log}^{2}n)$}.
\newblock {\em SIAM J. Comput.},
  \href{http://dx.doi.org/10.1137/0202004}{2:28--32}, 1973. \MR{0335268 (49
  \#50)}

\bibitem[Tak13]{Takaoka12}
T.~Takaoka.
\newblock A simplified algorithm for the all pairs shortest path problem with
  {$O(n^2\log n)$} expected time.
\newblock {\em J. Comb.\ Optim.},
  \href{http://dx.doi.org/10.1007/s10878-012-9550-3}{25(2):326--337}, 2013.
  \MR{3027629}

\bibitem[TH10]{MR2783521}
T.~Takaoka and M.~Hashim.
\newblock A simpler algorithm for the all pairs shortest path problem with
  {$O(n^2\log n)$} expected time.
\newblock In {\em Combinatorial optimization and applications. {P}art {II}},
  Lecture Notes in Comput.\ Sci.\
  \#6509\href{http://dx.doi.org/10.1007/978-3-642-17461-2_16}{, pages
  195--206}. Springer, 2010. \MR{2783521 (2012i:05319)}

\bibitem[Tho99]{Thorup99}
M.~Thorup.
\newblock Undirected single-source shortest paths with positive integer weights
  in linear time.
\newblock {\em J. ACM},
  \href{http://dx.doi.org/10.1145/316542.316548}{46(3):362--394}, 1999.
  \MR{1815738 (2001m:68145)}

\bibitem[Tho07]{Thorup07}
M.~Thorup.
\newblock Equivalence between priority queues and sorting.
\newblock {\em J. ACM},
  \href{http://dx.doi.org/10.1145/1314690.1314692}{54(6):Art.~28, 27}, 2007.
  \MR{2374029 (2008m:68026)}

\bibitem[TM80]{TaMo80}
T.~Takaoka and A.~Moffat.
\newblock An {$O(n^{2}\log n\log\log n)$} expected time algorithm for the all
  shortest distance problem.
\newblock In {\em Mathematical foundations of computer science}, Lecture Notes
  in Comput.\ Sci.\ \#88, pages 643--655. Springer, 1980. \MR{599597
  (82a:68087)}

\bibitem[Wil64]{Williams64}
J.~W.~J. Williams.
\newblock Algorithm 232: {Heapsort}.
\newblock {\em Comm.\ ACM},
  \href{http://dx.doi.org/10.1145/512274.512284}{7(6):347--348}, 1964.

\end{thebibliography}

\end{document}